\documentclass[11pt]{article}
\usepackage{amsmath}
\usepackage{amsthm}
\usepackage{amssymb}
\usepackage{algorithm}
\usepackage{subfig}
\usepackage{array}
\usepackage{multirow}
\usepackage{color}
\usepackage[english]{babel}
\usepackage{graphicx}
\usepackage{wrapfig,epsfig}
\usepackage{epstopdf}
\usepackage{url}
\usepackage{graphicx}
\usepackage{color}
\usepackage{epstopdf}
\usepackage{algpseudocode}
\usepackage{scrextend}
\usepackage[T1]{fontenc}
\usepackage{bbm}
\usepackage{comment}
\usepackage{multicol}
\usepackage{enumitem}
\setlist[enumerate]{topsep=0pt,itemsep=-1ex,partopsep=1ex,parsep=1ex}
\setlist[itemize]{topsep=0pt,itemsep=-1ex,partopsep=1ex,parsep=1ex}

\usepackage{tikz}
\usepackage{hyperref}  
\hypersetup{colorlinks=true,citecolor=red,linkcolor=blue} 
\usetikzlibrary{arrows}
\usepackage[margin=1in]{geometry}
\graphicspath{{./figs/}}

\usepackage{calc}

\usepackage{accents}

\author{Michael B. Cohen\thanks{Work was done while 
the first two authors were visiting Microsoft Research and hosted by S\'{e}bastien Bubeck,
and the third author was visiting University of Washington.
The authors would like to express their sincere gratitude to Rasmus Kyng for his questions about inverse maintenance that initiated this project. A preliminary version of this paper appeared in the Proceedings of 51th Annual ACM Symposium on Theory of Computing (STOC 2019). The full version of this paper appeared in the Journal of the Association for Computing Machinery (JACM 2020).
}
\\
MIT \& Microsoft Research\\
micohen@mit.edu\and Yin Tat Lee\footnotemark[1]
\\
UW \& Microsoft Research\\
yintat@uw.edu\and  Zhao Song\footnotemark[1]
\\
UT-Austin \& UW\\
zhaos@utexas.edu}

\date{}

\title{Solving Linear Programs in \\ the Current Matrix Multiplication Time}

\newtheorem{theorem}{Theorem}[section]
\newtheorem{lemma}[theorem]{Lemma}

\newtheorem{corollary}[theorem]{Corollary}

\newtheorem{assumption}[theorem]{Assumption}

\newtheorem{fact}[theorem]{Fact}
\newtheorem{remark}[theorem]{Remark}
\newtheorem{claim}[theorem]{Claim}

\newcommand{\wh}{\widehat}
\newcommand{\wt}{\widetilde}
\newcommand{\ov}{\overline}
\newcommand{\eps}{\epsilon}

\newcommand{\R}{\mathbb{R}}

\renewcommand{\varepsilon}{\epsilon}
\renewcommand{\tilde}{\wt}
\renewcommand{\hat}{\wh}
\renewcommand{\eps}{\epsilon}
\newcommand{\diag}{\textrm{diag}}

\DeclareMathOperator*{\E}{{\bf {E}}}
\DeclareMathOperator*{\Var}{{\bf {Var}}}
\DeclareMathOperator*{\Sup}{{\bf {Sup}}}

\DeclareMathOperator{\nnz}{nnz}

\DeclareMathOperator{\new}{new}

\DeclareMathOperator{\OPT}{OPT}

\global\long\def\defeq{\stackrel{\mathrm{{\scriptscriptstyle def}}}{=}}

\algnewcommand{\LineComment}[1]{\State \(\triangleright\) #1}

\makeatletter
\newcommand*{\RN}[1]{\expandafter\@slowromancap\romannumeral #1@}
\makeatother

\usepackage{lineno}

\usepackage{etoolbox}

\makeatletter

\newcommand{\define}[4][ignore]{%
  \ifstrequal{#1}{ignore}{}{
  \@namedef{thmtitle@#2}{#1}}%
  \@namedef{thm@#2}{#4}%
  \@namedef{thmtypen@#2}{lemma}%
  \newtheorem{thmtype@#2}[theorem]{#3}%
  \newtheorem*{thmtypealt@#2}{#3~\ref{#2}}%
}

\newcommand{\state}[1]{%
  \@namedef{curthm}{#1}
  \@ifundefined{thmtitle@#1}{
  \begin{thmtype@#1}
    }{
  \begin{thmtype@#1}[\@nameuse{thmtitle@#1}]
  }
    \label{#1}
    \@nameuse{thm@#1}
  \end{thmtype@#1}
  \@ifundefined{thmdone@#1}{
  \@namedef{thmdone@#1}{stated}%
  }{}
}

\newcommand{\restate}[1]{%
  \@namedef{curthm}{#1}
  \@ifundefined{thmtitle@#1}{
    \begin{thmtypealt@#1}
    }{
  \begin{thmtypealt@#1}[\@nameuse{thmtitle@#1}]
  }
    \@nameuse{thm@#1}
  \end{thmtypealt@#1}
  \@ifundefined{thmdone@#1}{
  \@namedef{thmdone@#1}{stated}%
  }{}
}

\newcommand{\thmlabel}[1]{
  \@ifundefined{thmdone@\@nameuse{curthm}}{\label{#1}
    }{\tag*{\eqref{#1}}}
}
\makeatother

\begin{document}

\begin{titlepage}
  \maketitle
  \begin{abstract}

This paper shows how to solve linear programs of the form
$\min_{Ax=b,x\geq0} c^\top x$ with $n$ variables
in time $$O^*((n^{\omega}+n^{2.5-\alpha/2}+n^{2+1/6}) \log(n/\delta))$$ where $\omega$ is the exponent of matrix multiplication,
$\alpha$ is the dual exponent of matrix multiplication, and $\delta$ is the relative accuracy.
For the current value of $\omega\sim2.37$ and $\alpha\sim0.31$, our algorithm takes $O^*(n^{\omega} \log(n/\delta))$ time.
When $\omega = 2$, our algorithm takes $O^*(n^{2+1/6} \log(n/\delta))$ time.

Our algorithm utilizes several new concepts that we believe may be
of independent interest: 
\begin{itemize}
\item We define a stochastic central path method.
\item We show how to maintain a projection matrix $\sqrt{W}A^{\top}(AWA^{\top})^{-1}A\sqrt{W}$
in sub-quadratic time under $\ell_{2}$ multiplicative changes in
the diagonal matrix $W$.
\end{itemize}

  \end{abstract}
  \thispagestyle{empty}
\end{titlepage}
%

\section{Introduction}

Linear programming is one of the key problems in computer science.
In both theory and practice, many problems can be reformulated as
linear programs to take advantage of fast algorithms. For an arbitrary
linear program $\min_{Ax=b,x\geq0}c^{\top}x$ with $n$ variables and
$d$ constraints\footnote{Throughout this paper, we assume there is no redundant constraints and hence $n\geq d$. 
Note that papers in different communities uses different symbols to denote the number of variables and constraints in a linear program.}, the fastest algorithm takes $O^*(\sqrt{d}\cdot\nnz(A)+d^{2.5})$\footnote{We use $O^*$ to hide $n^{o(1)}$ and $\log^{O(1)}(1/\delta)$ factors and $\tilde{O}$ to hide $\log^{O(1)}(n/\delta)$ factors.} 
where $\nnz(A)$ is the number of non-zeros in $A$ \cite{ls14,ls15}.

For the generic case $d=\Omega(n)$ we focus in this paper, the current fastest
runtime is dominated by $O^*(n^{2.5})$. This runtime has not been improved since a result by Vaidya on 1989 \cite{v87,v89b}. The
$n^{2.5}$ bound originated from two factors:  the cost per iteration
$n^{2}$ and the number of iterations $\sqrt{n}$. The $n^{2}$ cost
per iteration looks optimal because this is the cost to compute $Ax$ for a dense $A$. Therefore, many efforts  \cite{k84,r88,na89,v89a,ls14} have been focused on decreasing  
the number of iterations while maintaining the cost per iteration.
As for many important linear programs (and convex programs), the number
of iterations has been decreased, including maximum flow \cite{m13,m16}, minimum
cost flow \cite{cmsv17}, geometric median \cite{clmps16}, matrix scaling and balancing \cite{cmtv17}, and $\ell_{p}$ regression \cite{bcll18}. Unfortunately, beating
$\sqrt{n}$ iterations (or $\sqrt{d}$ when $d \ll n$) for the general case remains one of the biggest
open problems in optimization.

Avoiding this open problem, this paper develops a stochastic central path method that
has a runtime of $O^*(n^{\omega} + n^{2.5-\alpha/2} + n^{2+1/6})$, where $\omega$ is the exponent of matrix multiplication and $\alpha$ is the dual exponent of matrix multiplication\footnote{The dual exponent of matrix multiplication $\alpha$ is the supremum among all $a\geq0$ such that it takes $n^{2+o(1)}$ time to multiply an $n \times n$ matrix by an $n \times n^a$ matrix.}. For the current value of $\omega \sim 2.38$ and $\alpha \sim 0.31$, the runtime is simply $O^*(n^{\omega})$. This achieves a natural barrier for solving linear programs because linear systems are a special case of linear program and this is the best known runtime for solving linear systems. Although \cite{afl15, aw18, alman2019limits} showed that the exact and similar\footnote{Improving the matrix multiplication constant boils down to constructing/analyzing the tensors in better sense. Those work about limitation of matrix multiplication constant explore the exact same tensor and also variation of tensor in the previous work. For more details, we refer the readers to matrix multiplication literatures.} approaches used in \cite{cw87,w12,ds13,l14} cannot give a bound on $\omega$ better than $2.168$, we believe improving the additive $2+1/6$ term remains important for understanding linear programming. A recent work \cite{jswz20} improved the $2+1/6$ term to $2+1/18$.

Our method is a stochastic version of the short step central path method. This short step method takes $O^*(\sqrt{n})$ steps and each step decreases $x_i s_i$ by a $1-1/\sqrt{n}$ factor for all $i$ where $x$ is the primal variable and $s$ is the dual variable \cite{r88} (See the definition of $s$ in (\ref{eq:KKT})). This results in $O^*(\sqrt{n})\times n = O^*(n^{1.5})$ coordinate updates. Our method takes the same number of steps but only updates $\tilde{O}(\sqrt{n})$ coordinates each step. Therefore, we only update $O^*(n)$ coordinates in total, which is nearly optimal.

Our framework is efficient enough to take a much smaller step while maintaining the same running time. For the current value of $\omega\sim 2.38$, we show how to obtain the same runtime of $O^*(n^\omega)$ by taking $O^*(n)$ steps and $\tilde{O}(1)$ coordinates update per steps. This is because the complexity of each step decreases proportionally when the step size decreases. Beyond the cost per iteration, we remark that our algorithm is one of the very few central path algorithms \cite{prt02,m13,m16} that does not maintain $x_i s_i$ close to some ideal vector in $\ell_2$ norm. We are hopeful that our stochastic method and our proof will be useful for future research on interior point methods. 



\subsection{Related Work}

Interior point method has a long history, for more detailed surveys, we refer the readers to \cite{w97,y97,r01,rtv05,m12,t13}.
This paper is in part inspired by the use of data-structure in Laplacian solvers \cite{st04,kmp10,kmp11,ckmst11,kosz13,ckmpprx14,klpss16,ks16,ckkpprs18,kpsz18}, in particular the cycle update in \cite{kosz13}.

In a few recent follow-ups of this paper, the techniques developed in this work are generalized to a more broad class of optimization problems, i.e., Empirical Risk Minimization \cite{lsz19}, Cutting plane method \cite{jlsw20}, semi-definite programming \cite{jklps20}, deep neural network training \cite{bpsw20,clp+20}. A deterministic variant of our algorithm has been developed \cite{b20}, a sketching variant of our algorithm has been developed \cite{sy20}, a streaming variant has been developed \cite{lsz20}, the additive $1/6$ term has been improved to $1/18$ \cite{jswz20}, and the runtime has been improved to nearly linear time for dense linear programs with $n \gg d$ \cite{blss20}.

Matrix vector multiplication is a subtask of our iterative algorithm for solving linear programs. Online matrix-vector multiplication \cite{lw17,hkns15,ckl18} is closely related to our problem, but usually the computational model is different than our setting.

\section{Results and Techniques}

\begin{theorem}[Main result]\label{thm:main}
Given a linear program $\min_{Ax=b,x\geq0}c^{\top}x$ with no redundant
constraints. Assume that the polytope has diameter $R$ in $\ell_{1}$
norm, namely, for any $x\geq0$ with $Ax=b$, we have $\|x\|_{1}\leq R$.

Then, for any $0<\delta\leq1$, $\textsc{Main}(A,b,c,\delta)$ outputs
$x\geq0$ such that
\begin{align*}
c^{\top}x  \leq\min_{Ax=b,x\geq0}c^{\top}x+\delta\cdot\|c\|_{\infty}R \quad \text{and} \quad
\|Ax-b\|_{1} \leq\delta\cdot\left(R\sum_{i,j} | A_{i,j} |+\|b\|_{1}\right)
\end{align*}
in expected time
\begin{align*}
\left(n^{\omega+o(1)}+n^{2.5-\alpha/2+o(1)}+n^{2+1/6+o(1)}\right)\cdot\log(\frac{n}{\delta})
\end{align*}
where $\omega$ is the exponent of matrix multiplication, $\alpha$
is the dual exponent of matrix multiplication.

For the current value of $\omega \sim 2.38$ and $\alpha \sim 0.31$, the expected
time is simply $n^{\omega+o(1)}\log(\frac{n}{\delta})$.
\end{theorem}

See \cite{r88} and \cite[Sec E, F]{lee2013path} on the discussion on converting an approximation solution to an exact solution. For integral $A,b,c$, it suffices to pick $\delta = 2^{-O(L)}$ to get an exact solution where 
$L = \log(1+d_{\max} + \|c\|_\infty + \|b\|_\infty)$ is the bit complexity and $d_{\max}$ is the largest absolute value of the determinant of a square sub-matrix of $A$. For many combinatorial problems, $L=O(\log(n + \|b\|_\infty + \|c\|_\infty))$.

In this paper, we assume all floating point calculations are done exactly for simplicity. In general, the algorithm can be carried out with $O(L)$ bits of accuracy. See \cite{r88} for some discussions on the numerical stability of the interior point methods.

If $T(n)$ is the current cost of matrix multiplication and inversion with $T(n)\sim n^{2.38}$,
our runtime is simply $O(T(n)\log n\log(\frac{n}{\delta}))$.
The $\log(\frac{n}{\delta})$ comes from iteration count and
the $\log n$ factor comes from the doubling trick $( |y_{\pi(1.5r)}|\geq(1-1/\log n)|y_{\pi(r)}| )$ 
in the projection maintenance section. We left the problem of obtaining $O(T(n)\log(\frac{n}{\delta}))$ as an open problem.

Finally, we note that our runtime holds for any square and rectangular matrix multiplication algorithm as long as $\omega \leq 3 - \alpha$ (See Lemma \ref{lem:omega_leq_3_minus_a})\footnote{\cite{cglz20} proved a stronger result $\omega + 0.5\omega \alpha \leq 3$.} For example, Strassen algorithm together with a simple rectangular multiplication algorithm gives a runtime of roughly $n^{2.807}$.

\subsection{Central Path Method}
Our algorithm relies on two new ingredients: stochastic central path and projection maintenance. The central path method considers the linear programs
\begin{align*}
\min_{Ax=b,x\geq0} c^{\top}x \quad \text{(primal)} \quad \text{and} \quad \max_{A^{\top}y \leq c} b^{\top} y \quad \text{(dual)}
\end{align*}
with $A \in \R^{d \times n}$. Any solution of the linear program satisfies
the following optimality conditions:
\begin{align}
\label{eq:KKT}
x_{i}s_{i} & =0\text{ for all }i,\\
Ax & =b,\notag\\
A^{\top}y+s & =c,\notag\\
x_{i},s_{i} & \geq0\text{ for all }i.\notag
\end{align}
We call $(x,s,y)$ feasible if it satisfies the last three equations above.
For any feasible $(x,s,y)$, the duality gap of $(x,s,y)$ is $\sum_{i}x_{i}s_{i}$.
The central path method finds a solution of the linear
program by following the central path which uniformly decrease the duality gap. The central path $(x_{t},s_{t},y_{t})\in\R^{n+n+d}$ is a path parameterized by $t$ and defined by
\begin{align}\label{eq:primal_dual_central_path}
x_{t,i}s_{t,i} & =t\text{ for all }i,\\
Ax_{t} & =b,\notag\\
A^{\top}y_{t}+s_{t} & =c,\notag\\
x_{t,i},s_{t,i} & \geq0\text{ for all }i.\notag
\end{align}
where $x_{t,i}$ is the $i$-th coordinate of $x_t$ and $s_{t,i}$ is the $i$-th coordinate of $s_t$. It is known \cite{ytm94} how to transform linear programs by adding $O(n)$ many variables
and constraints so that:
\begin{itemize}
\item The optimal solution remains the same.
\item The central path at $t=1$ is near $(1_n,1_n,0_d)$ where $1_n$ and $0_d$
are all $1$ and all $0$ vectors with lengths $n$ and $d$.
\item It is easy to convert an approximate solution of the transformed program to the original one.
\end{itemize}
For completeness, a theoretical version of such result is included in Lemma~\ref{lem:feasible_LP}.
This result shows that it suffices to move gradually $(x_{1},s_{1},y_{1})$ to $(x_{t},s_{t},y_{t})$ for small enough $t$.

\subsubsection{Short Step Central Path Method}
The short step central path method maintains $x_{i}s_{i}=\mu_{i}$ for some vector $\mu$ such that
\begin{align}\label{eq:l2_invariant}
\sum_i (\mu_i - t)^2 = O(t^2) \quad \text{ for some scalar }t>0.
\end{align}
Since the duality gap is $\sum_i \mu_i$, it suffices to find $x$ and $s$ satisfying the above equation with small enough $t$. There are many variants of central path methods. We will focus on the version that decreases $t$ and takes a step of $\mu$ at the same time. The purpose of moving $\mu$ is to maintain the invariant (\ref{eq:l2_invariant}) and the purpose of decreasing $t$ is decrease the duality gap, which is roughly $nt$.
One natural way to maintain the invariant (\ref{eq:l2_invariant}) is to do a gradient descent step on the energy $\sum_i (\mu_i - t)^2$ defined in (\ref{eq:l2_invariant}), namely, moving $\mu$ to $\mu - h (\mu-t)$ with step size $h$\footnote{The classical view of central path method is to take a Newton step on the system (\ref{eq:primal_dual_central_path}), which turns out to be same as taking a gradient step on the energy defined in (\ref{eq:l2_invariant}). However, our main algorithm will choose a different energy and this gradient descent view is crucial for designing our algorithm.}. 

More generally, say we want to move from $\mu$ to
$\mu+\delta_{\mu}$, we approximate
the term $(x+\delta_{x})_i(s+\delta_{s})_i$ by $x_i s_i+x_i\delta_{s,i}+s_i\delta_{x,i}$
and obtain the following system:
\begin{align}
X\delta_{s}+S\delta_{x} & =\delta_{\mu},\notag\\
A\delta_{x} & =0,\label{eq:delta_x_s_y_mu}\\
A^{\top}\delta_{y}+\delta_{s} & =0,\notag
\end{align}
where $X=\diag(x)$ and $S=\diag(s)$. This equation is the linear approximation of the
original goal (moving from $\mu$ to $\mu+\delta_{\mu}$), and
that the step is explicitly given by the formula
\begin{align}\label{eq:d_step}
\delta_{x}=\frac{X}{\sqrt{XS}}(I-P)\frac{1}{\sqrt{XS}}\delta_{\mu}\text{ and }\delta_{s}=\frac{S}{\sqrt{XS}}P\frac{1}{\sqrt{XS}}\delta_{\mu},
\end{align}
where $P=\sqrt{\frac{X}{S}}A^{\top}\left(A\frac{X}{S}A^{\top}\right)^{-1}A\sqrt{\frac{X}{S}}$
is an orthogonal projection and the formulas $\frac{X}{\sqrt{XS}}, \frac{X}{S}, \cdots$ are the diagonal matrices of the corresponding vectors.

In turns out that one can decrease $t$ by $1-\frac{1}{\sqrt{n}}$ a multiplicative factor every iteration while maintaining the invariant (\ref{eq:l2_invariant}). This requires $\tilde{O}(\sqrt{n})$ iterations to converge. Combining this with the inverse maintenance technique \cite{v87}, this gives a total runtime of $n^{2.5}$.
More precisely, the algorithm maintains the invariant $\sum_i (\mu_i - t)^2 = O(t^2)$ by making steps bring $\mu_i$ closer to $t$ while taking steps to decrease $\mu_i$ uniformly. The progress of the whole algorithm is measured by $t$ because the duality gap of $(x_t,s_t,y_t)$ is bounded by $n t$. 

\subsubsection{Stochastic Central Path Method}
This part discusses how to modify the short step central path to decrease the cost per iteration to roughly $n^{\omega-\frac{1}{2}}$. 
Since our goal is to implement a central path method in sub-quadratic time per iteration, we do not even have the budget to compute $Ax$ every iterations. Therefore, instead of maintaining $\left(A\frac{X}{S}A^{\top}\right)^{-1}$ as shown in previous papers, we will study the problem of maintaining
a projection matrix $P=\sqrt{\frac{X}{S}}A^{\top}\left(A\frac{X}{S}A^{\top}\right)^{-1}A\sqrt{\frac{X}{S}}$ due to the formula of $\delta_x$ and $\delta_s$ (\ref{eq:d_step}). 

However, even if the projection matrix $P$ is given explicitly for free, it is difficult
to multiply the dense projection matrix with a dense vector $\delta_\mu$ in time $o(n^{2})$.
To avoid moving along a dense $\delta_{\mu}$, we move along an $O(k)$ sparse direction $\wt{\delta}_{\mu}$ defined by
\begin{align}
\wt{\delta}_{\mu,i}=\begin{cases}
\delta_{\mu,i}/p_{i} , & \text{with probability }p_{i}\defeq k\cdot\left(\frac{\delta_{\mu,i}^{2}}{ \sum_{l}\delta_{\mu,l}^{2} }+\frac{1}{n}\right);\\
0 , & \text{else}.
\end{cases}\label{eq:tilde_delta_mu}
\end{align}
The sparse direction is defined so that we are moving
in the same direction in expectation ($\E [ \wt{\delta}_{\mu,i} ] = \delta_{\mu,i}$)
and that the direction has as small variance as possible ($\E [ \wt{\delta}_{\mu,i}^{2} ] \leq\frac{\sum_{i}\delta_{\mu,i}^{2}}{k}$).
If the projection matrix is given explicitly, we can apply the projection matrix on $\wt{\delta}_{\mu}$ in time $O(nk)$.
This paper picks $k\sim\sqrt{n}$ and the sum of the cost of
projection vector multiplications in the whole algorithm is about $n k^2 = n^{2}$.

During the whole algorithm, we maintain a projection matrix 
$$\overline{P}=\sqrt{\frac{\overline{X}}{\overline{S}}}A^{\top}\left(A\frac{\overline{X}}{\overline{S}}A^{\top}\right)^{-1}A\sqrt{\frac{\overline{X}}{\overline{S}}}$$ for vectors $\overline{x}$ and $\overline{s}$ such that $\overline{x}_{i}$ and $\overline{s}_{i}$ are multiplicative approximations of $x_{i}$ and $s_{i}$ respectively for all $i$. Since we maintain the projection at a nearby point $(\overline{x}, \overline{s})$, our stochastic step $x\leftarrow x+\wt{\delta}_{x}$,
$s\leftarrow s+\wt{\delta}_{s}$ and $y\leftarrow y+ \wt{\delta}_{y}$ are defined by
\begin{align}
\overline{X}\wt{\delta}_{s}+\overline{S}\wt{\delta}_{x} & =\wt{\delta}_{\mu},\notag\\
A \wt{\delta}_{x} & =0,\label{eq:tilde_delta_x_s_y_mu}\\
A^{\top}\wt{\delta}_{y}+ \wt{\delta}_{s} & =0,\notag
\end{align}
which is different from (\ref{eq:delta_x_s_y_mu}) on both sides of the first equation.
Note that this system uses $\overline{X}$ and $\overline{S}$ because we have only maintained this projection matrix. The main goal of Section \ref{sec:stochastic} is to show $\overline{X} = \Theta(X)$ and $\overline{S} = \Theta(S)$ is good enough for our interior point method.
 Similar to (\ref{eq:d_step}), Lemma \ref{lem:stochastic_step_formula} shows that
\begin{align}\label{eq:d_step2}
\wt{\delta}_{x}=\frac{\overline{X}}{\sqrt{\overline{X}\overline{S}}}(I-\overline{P})\frac{1}{\sqrt{\overline{X}\overline{S}}}\wt{\delta}_{\mu}\text{ and }\wt{\delta}_{s}=\frac{\overline{S}}{\sqrt{\overline{X}\overline{S}}}\overline{P}\frac{1}{\sqrt{\overline{X}\overline{S}}}\wt{\delta}_{\mu}.
\end{align}

The previously fastest algorithm involves maintaining the matrix inverse $( A \frac{X}{S} A^\top )^{-1}$ using subspace embedding techniques \cite{s06,cw13,nn13} and leverage score sampling \cite{ss11}. In this paper, we maintain the projection directly using lazy updates.

The key departure from the central path we present is that we can only maintain 
$$0.9 t \leq \mu_i = x_{i}s_{i} \leq 1.1 t \quad \text{ for some }t>0$$
instead of $\mu$ close to $t$ in $\ell_2$ norm. We will
further explain the proof in Section~\ref{sec:stochastic_outline}.

\subsection{Projection Maintenance via Lazy Update}

The projection matrix we maintain is of the form $\sqrt{W}A^{\top}\left(AWA^{\top}\right)^{-1}A\sqrt{W}$
where $W=\diag(x/s)$. For intuition, we only explain how to maintain the matrix $M_{w} \defeq A^{\top}(AWA^{\top})^{-1}A$ for the short step central path step here.
In this case, we have $\sum_{i}\left(\frac{w_{i}^{\new}-w_{i}}{w_{i}}\right)^{2}=O(1)$ for each step. Given this, there are mainly two extreme cases, $w$ changes uniformly on all coordinates and $w$ changes only on a few coordinates.

If the changes $\left(\frac{w_{i}^{\new}-w_{i}}{w_{i}}\right)^{2}$ is uniform across all the coordinates, then $w_{i}^{\new}=(1\pm\frac{1}{\sqrt{n}})w_{i}$ for all $i$.
Since it takes $\sqrt{n}$ steps to change all coordinates
by a constant factor and we only need to maintain $M_{v}$ for some
$v_{i}=\Theta(w_{i})$ for all $i$, we can update the matrix
every $\sqrt{n}$ steps. Hence, the average cost per iteration of maintaining the
projection matrix is $n^{\omega-\frac{1}{2}}$, which is exactly what
we desired. 


For the other extreme case when $w$ changes on only a few coordinates, only $\sqrt{n}$ coordinates 
are changed by a constant factor during all $\sqrt{n}$ iterations. In
this case, instead of updating $M_{w}$ every step, we can compute
$M_{w}h$ online by the Woodbury matrix identity.

\begin{fact}[\cite{woodbury1950inverting}]\label{fact:woodbury}
The Woodbury matrix identity is 
\begin{align*}
( M + U C V )^{-1} = M^{-1} - M^{-1} U ( C^{-1} + V M^{-1} U )^{-1} V M^{-1}.
\end{align*}
\end{fact}

Let $S \subset [n]$ denote the set of coordinates that is changed by more than a constant factor and $r = |S|$.
Using the identity above, we have that
\begin{align} \label{eq:mw_update}
M_{w^{\new}} = M_{w} - (M_w)_S ( \Delta_{S,S}^{-1} + (M_w)_{S,S} )^{-1} ( (M_w)_S)^\top,
\end{align}
where $\Delta=\diag(w^{\new}-w)$, $(M_w)_S\in \R^{n \times r}$ is the $r$ columns from $S$ of $M_w$ and $(M_w)_{S,S},\Delta_{S,S} \in \R^{r \times r}$ are the $r$ rows and columns from $S$ of $M_w$ and $\Delta$.

As long as there are only few coordinates violating $v_{i}=\Theta(w_{i})$, (\ref{eq:mw_update}) can be applied online efficiently. In another case, we can use (\ref{eq:mw_update}) instead to
 update the matrix $M_{w}$ and
the cost is dominated by multiplying a $n\times n$ matrix with a
$n\times r$ matrix. 

\begin{theorem}[Rectangular matrix multiplication, \cite{gu18}]\label{lem:rmw}
Let the dual exponent of matrix multiplication $\alpha$ be the supremum among all $a\geq0$ such that it takes $n^{2+o(1)}$ time to multiply an $n \times n$ matrix by an $n \times n^a$ matrix.

Then, for any $n\geq r$, multiplying an $n \times r$ with an $r \times n$ matrix or $n \times n$ with $n \times r$ takes time
$$n^{2+o(1)}+r^{\frac{\omega-2}{1-\alpha}}n^{2-\frac{\alpha(\omega-2)}{1-\alpha}+o(1)}.$$
Furthermore, we have $\alpha>0.31389$.
\end{theorem}

See Lemma \ref{lem:rectangular_matrix_multiplication} for the origin of the formula.
Since the cost of multiplying $n\times n$ matrix by a $n\times1$
matrix is same as the cost for $n\times n$ with $n\times n^{0.31}$,
(\ref{eq:mw_update}) should be used to update at least $n^{0.31}$ coordinates. In
the extreme case only few $w_i$ are changing, we only need to update the matrix
$n^{\frac{1}{2}-0.31}$ times during the whole algorithm and each takes $n^{2}$ time, and hence
the total cost is less than $n^{\omega}$ for the current value of $\omega \sim 2.37$.

In previous papers \cite{k84,v89b,na91,nn94,ls14,ls15}, the matrix is updated in a fixed schedule independent of the input sequence $w$. This leads to sub-optimal bounds if used in this paper. We instead define a potential function to measure the distance between the approximate vector $v$ and the target vector $w$. 
When there are less than $n^{\alpha}$
coordinates of $v$ that is far from $w$, we are lazy and do not update the matrix. We simply apply the Woodbury matrix identity online. 
When there are more than $n^{\alpha}$
coordinates, we update $v$ by a certain greedy step. As in the extreme cases, the worst case for our algorithm is when $w$ changes uniformly across all coordinates 
and hence
the worst case runtime is $n^{\omega-\frac{1}{2}}$ per iteration. We will further explain the potential in Section~\ref{sec:projection_outline}.
\section{Notations}

For notational convenience, we assume the number of variables $n \geq 10$ and there are no redundant constraints. In particular, this implies that the constraint matrix $A$ is full rank and $n\geq d$.

For a positive integer $n$, let $[n]$ denote the set $\{1,2,\cdots,n\}$.

For any function $f$, we define $\wt{O}(f)$ to be $f\cdot \log^{O(1)}(f)$. In addition to $O(\cdot)$ notation, for two functions $f,g$, we use the shorthand $f\lesssim g$ (resp. $\gtrsim$) to indicate that $f\leq C g$ (resp. $\geq$) for some absolute constant $C$.

We use $\sinh x$ to denote $\frac{ e^x - e^{-x} }{2}$ and $\cosh x$ to denote $\frac{ e^x + e^{-x} }{2}$.

For vectors $a,b \in \R^n$ and accuracy parameter $\epsilon \in (0,1)$, we use $a \approx_{\epsilon} b$ to denote that $(1-\epsilon) b_i \leq a_i \leq (1+\epsilon) b_i, \forall i \in [n]$. Similarly, for any scalar $t$, we use $a \approx_{\epsilon} t$ to denote that $ (1-\epsilon) t \leq a_i \leq (1+\epsilon) t, \forall i \in [n]$.

For a vector $x \in \R^n$ and $s \in \R^n$, we use $xs$ to denote a length $n$ vector with the $i$-th coordinate $(xs)_i$ is $x_i \cdot s_i$. Similarly, we extend other scalar operations to vector coordinate-wise.

Given vectors $x, s \in \R^n$, we use $X$ and $S$ to denote the diagonal matrix of those two vectors. We use $\frac{X}{S}$ to denote the diagonal matrix given $(\frac{X}{S})_{i,i} = x_{i} / s_{i}$. Similarly, we extend other scalar operations to diagonal matrix diagonal-wise. Note that matrix $\sqrt{ \frac{X}{S} } A^\top ( A \frac{X}{S} A^\top )^{-1} A \sqrt{ \frac{X}{S} }$ is an orthogonal projection matrix.

\newpage
\section{Stochastic Central Path Method}\label{sec:stochastic}

\subsection{Proof Outline} \label{sec:stochastic_outline}

\begin{figure}[t] 
  \centering
    \includegraphics[width=0.7\textwidth]{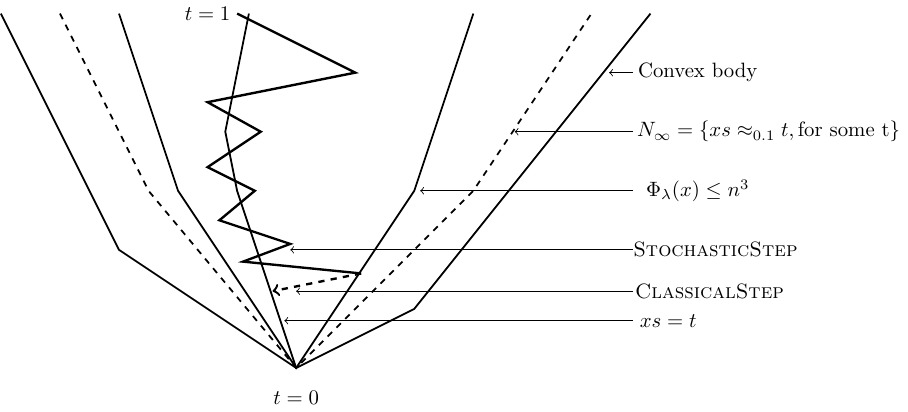}
    \caption{\textsc{ClassicalStep} happens with $n^{-2}$ probability\label{fig:step}}
\end{figure}

\begin{algorithm}[t]\caption{}\label{alg:stochastic_step}
\begin{algorithmic}[1]
\Procedure{\textsc{StochasticStep}}{$\mathrm{mp},x,s,\delta_{\mu},k,\epsilon$} \Comment{Lemma~\ref{lem:stochastic_step_formula},\ref{lem:stochastic_step},\ref{lem:bounding_mu_new_minus_mu}}
	\State $w \leftarrow \frac{x}{s}$, $\wt{v} \leftarrow \mathrm{mp}.\textsc{Update}(w)$ \Comment{Algorithm~\ref{alg:maintain_projection}}
	\State $\ov{x} \leftarrow x \sqrt{ \frac{ \wt{v} }{ w } }$, $\ov{s} \leftarrow s \sqrt{ \frac{ w }{ \wt{v} } }$ \Comment{It guarantees that $\frac{ \ov{x} }{ \ov{s} } = \wt{v}$ and $\ov{x} \ov{s} = x s$ }
	\Repeat
		\State Generate $\wt{\delta}_{\mu}$ such that \Comment{Compute a sparse direction}
		\State $\wt{\delta}_{\mu,i} \leftarrow \begin{cases} \delta_{\mu,i} / p_i , & \text{~with~prob.~} p_i = \min(1,k \cdot ( ( \delta_{\mu,i}^2 / \sum_{l=1}^n \delta_{\mu,l}^2 ) + 1/n ) ); \\ 0 & \text{~else.~} \end{cases}$
		\State \Comment{Compute an approximate step}
		\State {$\triangleright$ Find $(\wt{\delta}_x , \wt{\delta}_s , \wt{\delta}_y )$ such that these three equations hold \begin{align*} 
		\ov{X} \wt{\delta}_s + \ov{S} \wt{\delta}_x = & ~ \wt{\delta}_{\mu}, \\
		A \wt{\delta}_x = & ~ 0, \\
		A^\top \wt{\delta}_y + \wt{\delta}_s = & ~  0.
		 \end{align*} } \label{line:ds_satisfy}
		\State $p_{\mu} \leftarrow \mathrm{mp}.\textsc{Query}( \frac{1}{\sqrt{ \ov{X} \ov{S} } } \wt{\delta}_{\mu})$ \Comment{Algorithm~\ref{alg:maintain_projection}}
		\State $\wt{\delta}_s \leftarrow \frac{ \ov{S} }{ \sqrt{\ov{X} \ov{S} } } p_{\mu}$ \Comment{According to \eqref{eq:def_tilde_delta_s}}
		\State $\wt{\delta}_x \leftarrow \frac{1}{\ov{S}} \wt{\delta}_{\mu} - \frac{\ov{X}}{\sqrt{\ov{X} \ov{S}}} p_{\mu}$ \Comment{According to \eqref{eq:def_tilde_delta_x}}
	\Until { $\|{\overline{s}}^{-1} \wt{\delta}_s \|_{\infty} \leq \frac{1}{100\log{n}}$ and $ \|{\overline{x}}^{-1} \wt{\delta}_x \|_{\infty} \leq \frac{1}{100\log{n}}$ } \label{line:resample}
	\State \Return $( x + \wt{\delta}_x  ,  s + \wt{\delta}_s)$
\EndProcedure
\end{algorithmic}
\end{algorithm}

\begin{algorithm}\caption{Our main algorithm}\label{alg:main}
\begin{algorithmic}[1]
\Procedure{\textsc{Main}}{$A,b,c,\delta$} \Comment{Theorem~\ref{thm:main}}
	\State $\epsilon \leftarrow \frac{1}{40000\log{n}}$, $\epsilon_{\mathrm{mp}} \leftarrow \frac{1}{40000}$, $k \leftarrow \frac{1000 \epsilon \sqrt{n} \log^2 n}{\epsilon_{\mathrm{mp}}}$.
	\State $\lambda \leftarrow 40 \log n$, $\delta \leftarrow \min(\frac{\delta}{2}, \frac{1}{\lambda})$, $a \leftarrow \min(\alpha, 2/3)$.
	\State Modify the linear program and obtain an initial $x$ and $s$ according to Lemma~\ref{lem:feasible_LP}.
	\State $\textsc{MaintainProjection}~\mathrm{mp}$
	\State $\mathrm{mp}.\textsc{Initialize}(A,\frac{x}{s},\epsilon_{\mathrm{mp}}, a)$ \Comment{Algorithm~\ref{alg:maintain_projection}}
	\State $t \leftarrow 1$ \Comment{Initialize $t$}
	\While{$t > \delta^2/(32n^3)$} \Comment{We stop once the error is small enough}
		\State $t^{\new} \leftarrow ( 1 - \frac{ \epsilon }{3 \sqrt{n} } ) t$
		\State $\mu \leftarrow x s$
			\State $\delta_{\mu} \leftarrow ( \frac{t^{\new}}{t} - 1 ) x s - \frac{\epsilon}{2} \cdot t^{\new} \cdot \frac{ \nabla \Phi_{\lambda}( \mu / t - 1 ) }{ \| \nabla \Phi_{\lambda}(  \mu / t  - 1 ) \|_2 } $ \Comment{$\Phi_{\lambda}$ is defined in Lemma~\ref{lem:potential_function_Phi_cosh}}

			\State $(x^{\new} , s^{\new}) \leftarrow \textsc{StochasticStep}(\mathrm{mp},x,s,\delta_{\mu},k,\epsilon)$ \Comment{Algorithm~\ref{alg:stochastic_step}}
		\If{$\Phi_{\lambda} ( \mu^{\new} / t^{\new} - 1 ) > n^{3}$} \Comment{When potential function is large}
			\State $(x^{\new} , s^{\new}) \leftarrow \textsc{ClassicalStep}(x,s,t^{\new})$ \Comment{Lemma~\ref{lem:classical_step}, \cite{v89b}}
           \State $\mathrm{mp}.\textsc{Initialize}(A,\frac{x^{\new}}{s^{\new}},\epsilon_{\mathrm{mp}},a)$ \Comment{Restart the data structure}
		\EndIf
		\State $(x,s) \leftarrow (x^{\new},s^{\new})$, $t \leftarrow t^{\new}$ 
	\EndWhile
    \State Return an approximate solution of the original linear program according to Lemma \ref{lem:feasible_LP}.
\EndProcedure
\end{algorithmic}
\end{algorithm}

The short step central path method is defined using the approximation $(x+\delta_{x})_{i}(s+\delta_{s})_{i}\sim x_{i}s_{i}+x_{i}\delta_{s,i}+s_{i}\delta_{x,i}$.
This approximate is accurate if $\|X^{-1}\delta_{x}\|_{\infty} \leq 1/2$
and $\|S^{-1}\delta_{s}\|_{\infty} \leq 1/2$. For the $\delta_{x}$ step,
we have
\begin{align}\label{eq:stoc_cent_path_explain}
X^{-1}\delta_{x}=\frac{1}{\sqrt{XS}}(I-P)\frac{1}{\sqrt{XS}}\delta_{\mu}\sim\frac{1}{t}(I-P)\delta_{\mu},
\end{align}
where we used $x_{i}s_{i}\sim t$ for all $i$.

If we know that $\|\delta_{\mu}\|_{2}\leq t/4$, then the $\ell_{\infty}$
norm can be roughly bounded as follows:
\begin{align*}
\|X^{-1}\delta_{x}\|_{\infty}\leq\|X^{-1}\delta_{x}\|_{2}\lesssim\frac{1}{t}\|(I-P)\delta_{\mu}\|_{2}\leq\frac{1}{t}\|\delta_{\mu}\|_{2}\leq 1/2,
\end{align*}
where we used that $I-P$ is an orthogonal projection matrix. This
is the reason why a standard choice of $\delta_{\mu,i}$ is $-c t/\sqrt{n}$
for all $i$ for some small constant $c$.

For the stochastic step, $\tilde{\delta}_{\mu,i}\sim-\frac{t}{\sqrt{n}}\frac{n}{k}$
for roughly $k$ coordinates where the term $\frac{n}{k}$ is used to preserve the expectation of the step. Therefore, the $\ell_{2}$ norm of $\tilde{\delta}_{\mu}$
is very large ($\|\tilde{\delta}_{\mu}\|_{2}\sim t\sqrt{\frac{n}{k}}$).
After the projection, we have $\|X^{-1}\delta_{x}\|_{2}\sim\frac{1}{t}\|(I-P)\delta_{\mu}\|_{2}\sim\sqrt{\frac{n}{k}}$.
Hence, the bound of $\|X^{-1}\delta_{x}\|_{\infty}$ using $\|X^{-1}\delta_{x}\|_{2}$ is too weak.
To improve the bound, we use Chernoff bounds to estimate $\|X^{-1}\delta_{x}\|_{\infty}$. To simplify the proof, we use a loop in  Algorithm \ref{alg:stochastic_step} to ensure both the sup norm is always small not just with high probability.

Beside the $\ell_{\infty}$ norm bound, the proof sketch in (\ref{eq:stoc_cent_path_explain}) also
requires using $x_{i}s_{i}\sim t$ for all $i$. The short step central
path proof maintains an invariant that $\sum_{i}(x_{i}s_{i}-t)^{2}=O(t^{2})$.
However, since our stochastic step has a stochastic noise with $\ell_{2}$ norm
as large as $t\sqrt{\frac{n}{k}}$, one cannot hope to maintain $x_{i}s_{i}$
close to $t$ in $\ell_{2}$ norm. Instead, we follow an idea in \cite{ls14,lsw15}  and maintain the following potential

\begin{align*}
\sum_{i=1}^{n} \cosh\left(\lambda\left(\frac{x_{i}s_{i}}{t}-1\right)\right)=n^{O(1)}
\end{align*}
with $\lambda = \Theta(\log n)$. This potential is a variant of soft-max. Note that the potential bounded by
$n^{O(1)}$ implies that $x_{i} s_{i}$ is a multiplicative approximation
of $t$. To bound the potential, consider $r_{i} = \frac{ x_{i} s_{i} }{ t }$ and $\Phi(r)$ be the potential
above. Then, we have that
\begin{align*}
\E [ \Phi(r^{\new}) ] \leq \Phi(r)+\left\langle \nabla\Phi(r),\E [ r^{\new}-r ] \right\rangle + O(1) \E [ \| r^{\new} - r\|_{ \nabla^{2} \Phi(r) }^{2} ].
\end{align*}
The first order term can be bounded efficiently because $\E [ r^{\new}-r ]$
is close to the short step central path step. The second term is a variance
term which scales like $1/k$ due to the $k$ independent coordinates.
Therefore, the potential changed by $1/k\sim 1/\sqrt{n}$ factor each step.
Hence, we can maintain it for roughly $\sqrt{n}$ steps. 

To make sure the potential $\Phi$ is bounded during the whole algorithm, our step is
the mixtures of two steps of the form $\delta_{\mu}\sim-\frac{t}{\sqrt{n}}-t\frac{\nabla\Phi}{\|\nabla\Phi\|_{2}}$.
The first term is to decrease $t$ and the second term is to decrease
$\Phi$. 

Since the algorithm is randomized, there is a tiny probability that $\Phi$ is large. In that case, we switch to a short step central path method. See Figure \ref{fig:step}, Algorithm \ref{alg:stochastic_step}, and Algorithm \ref{alg:main}.  
The first part of the proof involves bounding every quantity listed in Table \ref{tab:notation}.
In the second part, we are using these quantities to bound the expectation
of $\Phi$.

To decouple the proof in both parts, we will make the following assumption in the first part. It will be verified in the second part.
\begin{assumption}\label{ass:x_s_mu}
Assume the following for the input of the procedure $\textsc{StochasticStep}$ (see Algorithm~\ref{alg:stochastic_step}):
\begin{itemize}
\item $xs \approx_{0.1} t$ with $t>0$.
\item $\mathrm{mp}.\textsc{Update}(w)$ outputs $\wt{v}$ such that $w \approx_{\epsilon_{\mathrm{mp}}} \wt{v}$ with $\epsilon_{\mathrm{mp}} \leq 1/40000$.
\item $\|\delta_\mu\|_2 \leq \epsilon t$ with $0<\epsilon<1/(40000 \log n)$.
\item $k \geq 1000 \epsilon \sqrt{n} \log^2{n} / \epsilon_{\mathrm{mp}}$.
\end{itemize}
\end{assumption}
The data structure $\mathrm{mp}$ in both Algorithm \ref{alg:stochastic_step} and Algorithm \ref{alg:main} is used to maintain some approximation of the projection matrix. It is formally defined in Section \ref{sec:projection_maintenance_datastructure}. For this section, the only facts we need is that $w \approx_{\epsilon_{\mathrm{mp}}} \wt{v}$ stated in the assumption and that the vector $\mathrm{mp}.\textsc{Query}(w)$ outputs satisfies line \ref{line:ds_satisfy} in Algorithm \ref{alg:stochastic_step}.
\subsection{Bounding each quantity of stochastic step}
First, we give an explicit formula for our step, which will be used in all subsequent calculations.

\begin{lemma}\label{lem:stochastic_step_formula}
The procedure $\textsc{StochasticStep}(\mathrm{mp}, x, s, \delta_{\mu}, k, \epsilon)$ (see Algorithm~\ref{alg:stochastic_step}) finds a solution $\wt{\delta}_x$, $\wt{\delta}_s \in \R^n$ to \eqref{eq:tilde_delta_x_s_y_mu} by the formula
\begin{align}
\wt{\delta}_x = & ~ \frac{ \ov{X} }{ \sqrt{ \ov{X} \ov{S} } } (I - \ov{P}) \frac{1}{ \sqrt{ \ov{X} \ov{S} } } \wt{\delta}_{\mu} \label{eq:def_tilde_delta_s} \\
\wt{\delta}_s = & ~ \frac{ \ov{S} }{ \sqrt{ \ov{X} \ov{S} } } \ov{P} \frac{1}{ \sqrt{ \ov{X} \ov{S} } } \wt{\delta}_{\mu} \label{eq:def_tilde_delta_x}
\end{align}
with
\begin{align}\label{eq:def_ov_P}
\ov{P} = \sqrt{ \frac{ \ov{X} }{ \ov{S} } } A^\top \left( A \frac{ \ov{X} }{ \ov{S} } A^\top \right)^{-1} A \sqrt{ \frac{ \ov{X} }{ \ov{S} } }.
\end{align}
\end{lemma}

\begin{table}[t]
\begin{center}
    \begin{tabular}{| l | l | l | l | l | l |}
    \hline
    Quantity & Bound   & Place   \\ \hline %
    $\| \E [ s^{-1} \wt{\delta}_{s} ] \|_2, \| \E [ x^{-1} \wt{\delta}_{x} ] \|_2, \| \E [ \mu^{-1} \wt{\delta}_{\mu} ] \|_2 $ & $O(\epsilon)$ & Part 1, Lemma~\ref{lem:stochastic_step} \\ \hline 
    $\| \E [ \mu^{-1} (\mu^{\new} - \mu - \wt{\delta}_{\mu}) ] \|_2$ & $O(\epsilon_{\mathrm{mp}} \cdot \epsilon)$ & Part 1, Lemma~\ref{lem:bounding_mu_new_minus_mu} \\ \hline
    $\| \E [ \mu^{-1} ( \mu^{\new} - \mu ) ] \|_2$ & $O(\epsilon)$ & Part 1, Lemma~\ref{lem:bounding_mu_new_minus_mu} \\ \hline
    $\Var[ s_i^{-1} \wt{\delta}_{s,i} ], \Var[ x_i^{-1} \wt{\delta}_{x,i} ], \Var[ \mu_i^{-1} \wt{\delta}_{\mu,i} ]$ & $O(\epsilon^2/k)$ & Part 2, Lemma~\ref{lem:stochastic_step} \\ \hline
    $\Var[ \mu_i^{-1} \mu^{\new} ]$ & $O(\epsilon^2/k)$ & Part 2, Lemma~\ref{lem:bounding_mu_new_minus_mu} \\ \hline
    $\| s^{-1} \wt{\delta}_{s} \|_{\infty}, \| x^{-1} \wt{\delta}_{x} \|_{\infty}, \| \mu^{-1} \wt{\delta}_{\mu} \|_{\infty}$ & $O(1/\log{n})$ & Part 3, Lemma~\ref{lem:stochastic_step} \\ \hline
    $\| \mu^{-1} ( \mu^{\new} - \mu ) \|_{\infty}$ & $O(1/\log{n})$ & Part 3, Lemma~\ref{lem:bounding_mu_new_minus_mu} \\ \hline
    \end{tabular}\caption{The bound of each quantity under Assumption~\ref{ass:x_s_mu}. For intuition, think $\epsilon \sim \epsilon_{\mathrm{mp}} \sim 1/10$ and $k \sim \sqrt{n}$.}\label{tab:notation}
\end{center}
\end{table}

\begin{proof}

For the first equation of  \eqref{eq:tilde_delta_x_s_y_mu}, we multiply $ A \ov{S}^{-1}$ on both sides,
\begin{align*}
A \ov{S}^{-1} \ov{X} \wt{\delta}_s + A \wt{\delta}_x = A \ov{S}^{-1} \wt{\delta}_{\mu}.
\end{align*}
Since the second equation gives $A \wt{\delta}_x = 0$, then we know that $A \ov{S}^{-1} \ov{X} \wt{\delta}_s = A \ov{S}^{-1} \wt{\delta}_{\mu}$.

Multiplying $A \ov{S}^{-1} \ov{X}$ on both sides of the third equation of  \eqref{eq:tilde_delta_x_s_y_mu}, we have
\begin{align*}
- A \ov{S}^{-1} \ov{X} A^\top \wt{\delta}_y = A \ov{S}^{-1} \ov{X} \wt{\delta}_s = A \ov{S}^{-1} \wt{\delta}_{\mu}. 
\end{align*}
Thus, 
\begin{align*}
\wt{\delta}_y = & ~ - (A \ov{S}^{-1} \ov{X} A^\top)^{-1} A \ov{S}^{-1} \wt{\delta}_{\mu}, \\
\wt{\delta}_s = & ~ A^\top (A \ov{S}^{-1} \ov{X} A^\top)^{-1} A \ov{S}^{-1} \wt{\delta}_{\mu},  \\
\wt{\delta}_x = & ~ \ov{S}^{-1} \wt{\delta}_{\mu} - \ov{S}^{-1} \ov{X} A^\top ( A \ov{S}^{-1} \ov{X} A^\top )^{-1} A \ov{S}^{-1} \wt{\delta}_{\mu}.
\end{align*}
Recall we define $\ov{P}$ as  \eqref{eq:def_ov_P}, then we have
\begin{align*}
\wt{\delta}_s = & ~ \frac{ \ov{S} }{ \sqrt{ \ov{X} \ov{S} } } \cdot \sqrt{ \frac{ \ov{X} }{ \ov{S} } } A^\top (A \frac{ \ov{X} }{ \ov{S} } A^\top )^{-1} \sqrt{ \frac{ \ov{X} }{ \ov{S} } } \cdot \frac{1}{ \sqrt{ \ov{X} \ov{S} } } \wt{\delta}_{\mu} 
= \frac{ \ov{S} }{ \sqrt{ \ov{X} \ov{S} } } \ov{P} \frac{1}{ \sqrt{ \ov{X} \ov{S} } } \wt{\delta}_{\mu},
\end{align*}
and
\begin{align*}
\wt{\delta}_x = & ~ \ov{S}^{-1} \wt{\delta}_{\mu} - \frac{ \ov{X} } { \sqrt{ \ov{X} \ov{S} } } \cdot \sqrt{ \frac{ \ov{X} }{ \ov{S} } } A^\top (A \frac{ \ov{X} }{ \ov{S} } A^\top )^{-1} \sqrt{ \frac{ \ov{X} }{ \ov{S} } } \cdot \frac{1}{ \sqrt{ \ov{X} \ov{S} } } \wt{\delta}_{\mu} 
=  \frac{ \ov{X} } { \sqrt{ \ov{X} \ov{S} } }  (I - \ov{P}) \frac{ 1 } { \sqrt{ \ov{X} \ov{S} } } \wt{\delta}_{\mu}.
\end{align*}
which are matching  \eqref{eq:def_tilde_delta_s} and \eqref{eq:def_tilde_delta_x}.

To see why the $\textsc{StochasticStep}$ outputs $\wt{\delta}_x$, $\wt{\delta}_s$ satisfying \eqref{eq:def_tilde_delta_s} and \eqref{eq:def_tilde_delta_x}, we note that  $$p_{\mu}=\sqrt{\widetilde{V}}A^{\top}\left(A\frac{\overline{X}}{\overline{S}}A^{\top}\right)^{-1}A\sqrt{\widetilde{V}}\frac{1}{\sqrt{\overline{X}\overline{S}}}\widetilde{\delta}_{\mu}=\overline{P}\frac{1}{\sqrt{\overline{X}\overline{S}}}\widetilde{\delta}_{\mu}$$ because of Theorem~\ref{thm:maintain_projection}.
\end{proof}
Using the explicit formula, we are ready to bound all quantities we needed in the following two subsubsections.

\subsubsection{Bounding $\wt{\delta}_s$, $\wt{\delta}_x$ and $\wt{\delta}_\mu$}
\begin{lemma}[]\label{lem:stochastic_step}
Under the Assumption~\ref{ass:x_s_mu}, the two vectors $\wt{\delta}_x$ and $\wt{\delta}_s$ found by $\textsc{StochasticStep}$ satisfy : \\
1. $\| \E[ \ov{s}^{-1} \wt{\delta}_s ]  \|_2 \leq 2 \epsilon , \| \E[ \ov{x}^{-1} \wt{\delta}_x ] \|_2 \leq 2 \epsilon, \| \E[ s^{-1} \wt{\delta}_s ]  \|_2 \leq 2 \epsilon , \| \E[ x^{-1} \wt{\delta}_x ] \|_2 \leq 2 \epsilon,
 \| \E [ \mu^{-1} \wt{\delta}_{\mu} ] \|_2 \leq 4 \epsilon $. \\
2.  $\Var[ \frac{\wt{\delta}_{s,i}}{\ov{s}_i}] \leq \frac{2 \epsilon^2}{k}, \Var[\frac{\wt{\delta}_{x,i}}{\ov{x}_i} ] \leq \frac{2 \epsilon^2}{k}, \Var[ \frac{\wt{\delta}_{s,i}}{s_i}] \leq \frac{2 \epsilon^2}{k}, \Var[\frac{\wt{\delta}_{x,i}}{x_i}] \leq \frac{2 \epsilon^2}{k}, \Var[ \frac{\wt{\delta}_{\mu,i}}{\mu_i} ] \leq \frac{8 \epsilon^2}{k}$. \\
3. $\|\overline{s}^{-1}\wt{\delta}_{s}\|_{\infty}\leq \frac{0.01}{\log{n}},\|s^{-1}\wt{\delta}_{s}\|_{\infty}\leq \frac{0.02}{ \log{n}},\|\overline{x}^{-1}\wt{\delta}_{x}\|_{\infty}\leq \frac{0.01}{ \log{n}},\|x^{-1}\wt{\delta}_{x}\|_{\infty}\leq \frac{0.02}{\log{n}},\|\mu^{-1}\wt{\delta}_{\mu}\|_{\infty}\leq \frac{0.02}{\log{n}}$.
\end{lemma}

\begin{remark}
For notational simplicity, the $\E$ and $\Var$ in the proof are for the case without resampling (Line \ref{line:resample}). Since the all the additional terms due to resampling are polynomially bounded and since we can set failure probability to an arbitrarily small inverse polynomial (see Claim~\ref{claim:prob}), if we took into account the extra variance from resampling, the proof does not change and the result remains the same.
\end{remark}

\begin{proof}

\begin{claim}[Part 1, bounding the $\ell_2$ norm of expectation]\label{cla:bounding_l2_norm_expectation}
\begin{align*}
\| \E[ \ov{s}^{-1} \wt{\delta}_s ]  \|_2 \leq 2 \epsilon , \| \E[ \ov{x}^{-1} \wt{\delta}_x ] \|_2 \leq 2 \epsilon, \| \E[ s^{-1} \wt{\delta}_s ]  \|_2 \leq 2 \epsilon , \| \E[ x^{-1} \wt{\delta}_x ] \|_2 \leq 2 \epsilon,
 \| \E [ \mu^{-1} \wt{\delta}_{\mu} ] \|_2 \leq 4 \epsilon .
 \end{align*}
\end{claim}
\begin{proof}

For $\| \ov{s}^{-1} \wt{\delta}_s \|_{\infty}$, we consider the $i$-th coordinate of the vector
\begin{align*}
\ov{s}^{-1}_i \wt{\delta}_{s,i} = \frac{ 1 }{ \sqrt{ \ov{x}_i \ov{s}_i } } \sum_{j=1}^n \ov{P}_{i,j} \frac{ \wt{\delta}_{\mu,j} }{ \sqrt{ \ov{x}_j \ov{s}_j } }.
\end{align*}
Then, we have
\begin{align*}
\E \left[ \ov{s}_i^{-1} \wt{\delta}_{s,i} \right] = \frac{ 1 }{ \sqrt{ \ov{x}_i \ov{s}_i } } \sum_{j=1}^n \ov{P}_{i,j} \frac{ \E [ \wt{\delta}_{\mu,j} ] }{ \sqrt{ \ov{x}_j \ov{s}_j } } = \frac{ 1 }{ \sqrt{ \ov{x}_i \ov{s}_i } } \sum_{j=1}^n \ov{P}_{i,j} \frac{  \delta_{\mu,j} }{ \sqrt{ \ov{x}_j \ov{s}_j } }.
\end{align*}

Since $x s \approx_{0.1} t$ and $\| \delta_{\mu} \|_2 \leq \epsilon t$, we have $\| \frac{ \delta_{\mu} }{ \sqrt{ x s } } \|_2 \leq \frac{1.1 \epsilon t}{ \sqrt{t} }$. Since $\ov{P}$ is an orthogonal projection matrix, we have $\| \ov{P} \frac{ \delta_{\mu} }{ \sqrt{ \ov{x} \ov{s} } } \|_2 \leq \| \frac{ \delta_{\mu} }{ \sqrt{ \ov{x} \ov{s} } } \|_2$. Putting all the above facts and $x s = \ov{x} \ov{s}$, we can show
\begin{align*}
\left\| \E [ \ov{s}^{-1} \wt{\delta}_s ] \right\|_2^2
= & ~ \sum_{i=1}^n \left( \frac{1}{ \sqrt{\ov{x}_i \ov{s}_i } } \sum_{j=1}^n \ov{P}_{i,j} \frac{ \delta_{\mu,j} }{ \sqrt{ \ov{x}_j \ov{s}_j } } \right)^2
= \sum_{i=1}^n \frac{1}{ \ov{x}_i \ov{s}_i } \left( \sum_{j=1}^n \ov{P}_{i,j} \frac{ \delta_{\mu,j} }{ \sqrt{ \ov{x}_j \ov{s}_j } } \right)^2 \\
\leq & ~ \frac{1}{0.9 t } \sum_{i=1}^n \left( \sum_{j=1}^n \ov{P}_{i,j} \frac{ \delta_{\mu,j} }{ \sqrt{ \ov{x}_j \ov{s}_j } } \right)^2
=  \frac{1}{0.9t} \| \ov{P} \frac{\delta_{\mu}}{ \sqrt{ \ov{x} \ov{s} } } \|_2^2 \\
\leq & ~ \frac{1}{0.9 t} \|  \frac{\delta_{\mu}}{ \sqrt{ \ov{x} \ov{s} } } \|_2^2 
\leq  \frac{(1.1)^2}{0.9 t} \cdot \frac{(\epsilon t)^2}{t}
\leq  1.4 \epsilon^2,
\end{align*}
which implies that
\begin{align}\label{eq:bars_delta_s_norm}
\left\| \E [ \ov{s}^{-1} \wt{\delta}_s ] \right\|_2 \leq 1.2 \epsilon.
\end{align}
Notice that the proof for $x$ is identical to the proof for $s$ because $(I - \ov{P})$ is also a projection matrix. Since $\ov{s} \approx_{0.1} s$ and $\ov{x} \approx_{0.1} x$, then we can also prove the next two inequalities in the Claim statement.

Now, we are ready to bound $\| \E[ \mu^{-1} \wt{\delta}_{\mu} ] \|_2$ 
\begin{align*}
\| \E[ \mu^{-1} \wt{\delta}_{\mu} ] \|_2
= \| \E[ \ov{s}^{-1} \ov{x}^{-1} ( \ov{x} \wt{\delta}_s + \ov{s} \wt{\delta}_x ) ]  \|_2
\leq \| \E[\ov{s}^{-1} \wt{\delta}_s ] \|_2 + \| \E[ \ov{x}^{-1} \wt{\delta}_x ] \|_2
\leq 4 \epsilon.
\end{align*}
by using $\mu = xs = \ov{x} \ov{s}$ and $\ov{x} \wt{\delta}_s + \ov{s} \wt{\delta}_x = \wt{\delta}_{\mu}$ from  \eqref{eq:tilde_delta_x_s_y_mu}.
\end{proof}

\begin{claim}[Part 2, bounding the variance per coordinate]\label{cla:bounding_coordinate_variance}
\begin{align*}
\Var[ \ov{s}_i^{-1} \wt{\delta}_{s,i}] \leq \frac{2 \epsilon^2}{k}, \Var[ \ov{x}_i^{-1} \wt{\delta}_{x,i} ] \leq \frac{2 \epsilon^2}{k}, \Var[ s_i^{-1} \wt{\delta}_{s,i}] \leq \frac{2 \epsilon^2}{k}, \Var[ x_i^{-1} \wt{\delta}_{x,i} ] \leq \frac{2 \epsilon^2}{k}, \Var[ \mu_i^{-1} \wt{\delta}_{\mu,i} ] \leq \frac{8 \epsilon^2}{k}.
\end{align*}
\end{claim}
\begin{proof}

Consider the $i$-th coordinate of the vector
\begin{align*}
\ov{s}^{-1}_i \wt{\delta}_{s,i} = \frac{ 1 }{ \sqrt{ \ov{x}_i \ov{s}_i } } \sum_{j=1}^n \ov{P}_{i,j} \frac{ \wt{\delta}_{\mu,j} }{ \sqrt{ \ov{x}_j \ov{s}_j } }.
\end{align*}

For variance of $\ov{s}_i^{-1} \wt{\delta}_{s,i}$, we have
\begin{align*}
\Var[ \ov{s}^{-1}_i \wt{\delta}_{s,i} ]
= & ~ \frac{1}{ \ov{x}_i \ov{s}_i } \sum_{j=1}^n \frac{ \ov{P}_{i,j}^2 }{ \ov{x}_j \ov{s}_j } \Var[ \wt{\delta}_{\mu,j} ] & \text{~by~all~}\wt{\delta}_{\mu,j}\text{~are~independent} \\
\leq & ~ \frac{1}{ \ov{x}_i \ov{s}_i } \sum_{j=1}^n \frac{ \ov{P}_{i,j}^2 }{ \ov{x}_j \ov{s}_j } \frac{1}{k} \frac{ \delta_{\mu,j}^2 }{ \frac{ \delta_{\mu,j}^2 }{ \sum_{l=1}^n \delta_{\mu,l}^2  } + \frac{1}{n} } & \text{~by~ \eqref{eq:tilde_delta_mu}}  \\
\leq & ~ \frac{1}{ \ov{x}_i \ov{s}_i } \sum_{j=1}^n \frac{ \ov{P}_{i,j}^2 }{ \ov{x}_j \ov{s}_j } \frac{1}{k} \sum_{l=1}^n \delta_{\mu,l}^2 \\
\leq & ~ \frac{ 1.3  }{ t^2 } \sum_{j=1}^n \ov{P}_{i,j}^2 \frac{1}{k} \sum_{l=1}^n \delta_{\mu,l}^2  \leq \frac{1.3 \epsilon^2}{k},
& \text{~by~}\ov{x}_i \ov{s}_i = x_i s_i \approx_{1/10} t
\end{align*}
where we used that $\sum_{j=1}^n \ov{P}_{i,j}^2 = \ov{P}_{i,i} \leq 1$, $\| \delta_{\mu} \|_2 \leq \epsilon t$ at the end.

The proof for the other three inequalities in the Claim statement are identical to this one. We omit here.

For the variance of $\mu_i^{-1} \wt{\delta}_{\mu,i}$, 
\begin{align}
\Var[ \mu_i^{-1} \wt{\delta}_{\mu,i} ] 
= & ~ \Var[ \ov{x}_i^{-1} \ov{s}_i^{-1} ( \ov{x}_i \wt{\delta}_{s,i} + \ov{s}_i \wt{\delta}_{x,i} )  ] \notag \\
\leq & ~ 2 \Var[  \ov{x}_i^{-1} \ov{x}_i \ov{s}_i^{-1} \wt{\delta}_{s,i} ] + 2 \Var[ \ov{s}_i^{-1} \ov{s}_i \ov{x}_i^{-1} \wt{\delta}_{x,i} ] \notag \\
= & ~ 2 \Var[ \ov{s}_i^{-1} \wt{\delta}_{s,i} ] + 2 \Var[ \ov{x}_i^{-1} \wt{\delta}_{x,i} ]
\leq 8 \epsilon^2 / k. \notag
\end{align} 
where we used the definition $\mu = xs = \ov{x} \ov{s}$ and \eqref{eq:tilde_delta_x_s_y_mu} in the first step, the triangle inequality in the second step and, $\Var[\ov{s}_i^{-1} \wt{\delta}_{s,i} ], \Var[\ov{x}_i^{-1} \wt{\delta}_{x,i} ] \leq 2 \epsilon^2/k$ at the end.
\end{proof}

\begin{claim}[Part 3, bounding the probability of success]\label{claim:prob}
Without resampling, the following holds with probability $1-2n \exp(-\frac{0.003 k}{\epsilon\sqrt{n}\log{n}})$. 
$$\|\overline{s}^{-1}\wt{\delta}_{s}\|_{\infty}\leq \frac{0.01}{\log{n}},\|s^{-1}\wt{\delta}_{s}\|_{\infty}\leq \frac{0.02}{\log{n}},\|\overline{x}^{-1}\wt{\delta}_{x}\|_{\infty}\leq \frac{0.01}{\log{n}},\|x^{-1}\wt{\delta}_{x}\|_{\infty}\leq \frac{0.02}{\log{n}},\|\mu^{-1}\wt{\delta}_{\mu}\|_{\infty}\leq \frac{0.02}{\log{n}}.$$
With resampling, it always holds.
\end{claim}
\begin{proof}
We can write $\ov{s}_i^{-1} \wt{\delta}_{s,i} - \E[ \ov{s}_i^{-1} \wt{\delta}_{s,i}] = \sum_j Y_j$  where $Y_j$ are independent random variables defined by
\begin{align*}
Y_j = \frac{1}{ \sqrt{ \ov{x}_i \ov{s}_i } } \ov{P}_{i,j} \frac{ \wt{\delta}_{\mu,j} }{ \sqrt{ \ov{x}_j \ov{s}_j } } - \frac{1}{ \sqrt{ \ov{x}_i \ov{s}_i } } \ov{P}_{i,j} \frac{ \delta_{\mu,j} }{ \sqrt{ \ov{x}_j \ov{s}_j } }.
\end{align*}
We bound the sum using Bernstein inequality. Note that $Y_j$ are mean $0$ and that Claim~\ref{cla:bounding_coordinate_variance} shows that
$\sum_{j=1}^n \E[ Y_j^2 ] = \Var[ \ov{s}_i^{-1} \wt{\delta}_{s,i}] \leq \frac{2 \epsilon^2}{k}.$
We also need to give an upper bound for $Y_j$
\begin{align*}
|Y_j| 
= & ~ \left| \frac{1}{ \sqrt{\ov{x}_i \ov{s}_i} } \ov{P}_{i,j} \left(\frac{ \wt{\delta}_{\mu,j} - \delta_{\mu,j} }{ \sqrt{\ov{x}_j \ov{s}_j } } \right)  \right|  \\
\leq & ~ \frac{1.2}{t} | \wt{\delta}_{\mu,j} - \delta_{\mu,j} |   & \text{~by~} |\ov{P}_{i,j}| \leq 1, x_i s_i \approx_{1/10} t \\
\leq & ~ \frac{1.2}{t} | \delta_{\mu,j} / p_j | & \text{~by~} \wt{\delta}_{\mu,j} \in [0, \delta_{\mu,j} / p_j ] \\
= & ~ \frac{1.2}{t} \frac{1}{ k } \frac{1}{ ( \frac{ \delta_{\mu,i} }{ \sum_{l=1}^n \delta_{\mu,l}^2 } + \frac{1}{ n \delta_{\mu,i} }  ) }  & \text{~by~ \eqref{eq:tilde_delta_mu}} \\
\leq & ~ \frac{0.6}{t} \frac{1}{ k } \left( n \sum_{l=1}^n \delta_{\mu,l}^2 \right)^{1/2}  & \text{~by~} a^2 + b^2 \geq 2ab \\
\leq & ~ \frac{0.6 \epsilon \sqrt{ n } }{ k } \defeq M. & \text{~by~} \| \delta_{\mu} \|_2 \leq \epsilon t
\end{align*}

Now, we can apply Bernstein inequality 
\begin{align*}
\Pr \left[ \left| \sum_{j=1}^n Y_j \right| > b \right] \leq & ~ 2\exp \left( - \frac{ b^2/2 }{ \sum_{j=1}^n \E[Y_j^2] + M b/3 } \right)  \\
\leq & ~ 2 \exp \left( - \frac{ b^2 / 2 }{ 2 \epsilon^2 / k + ( 0.6 \epsilon \sqrt{n} / k ) \cdot b /3 } \right).
\end{align*}
We choose $b = \frac{0.005}{\log{n}}$ and use $\epsilon \leq \frac{1}{400\log{n}}$ and $n\geq 10$ to get
\begin{align*}
\Pr \left[ \left| \sum_{j=1}^n Y_j \right| \geq \frac{0.05}{\log{n}} \right] \leq 2\exp \left( -\frac{0.003 k}{\epsilon\sqrt{n}\log{n}} \right).
\end{align*}

Since $\|\E[ \ov{s}_i^{-1} \wt{\delta}_{s,i}] \|_2 \leq 2 \eps \leq \frac{0.005}{\log{n}}$, we have that $|\ov{s}_i^{-1} \wt{\delta}_{s,i}|\leq \frac{0.01}{\log{n}}$ with probability $1-2\exp(-\frac{0.003 k}{\epsilon\sqrt{n}\log{n}})$.
Taking {\color{red}a} union bound, we have that $\|\ov{s}^{-1} \wt{\delta}_{s}\|_\infty\leq \frac{0.01}{\log{n}}$ with probability $1-2n \exp(-\frac{0.003 k}{\epsilon\sqrt{n}\log{n}})$. Similarly, this holds for the other 3 terms.

Now, the last term follows by the calculation
\begin{align*}
| \mu_i^{-1} \wt{\delta}_{\mu,i} | 
=  | \ov{x}_i^{-1} \ov{s}_i^{-1} ( \ov{x}_i \wt{\delta}_{s,i} + \ov{s}_i \wt{\delta}_{x,i} ) | 
= | \ov{s}_i^{-1} \wt{\delta}_{s,i} | + | \ov{x}_i^{-1} \wt{\delta}_{x,i} | 
\leq \frac{0.02}{\log{n}}.
\end{align*}
\end{proof}
\end{proof}


\subsubsection{Bounding $\mu^{\new} - \mu$}

\begin{lemma}\label{lem:bounding_mu_new_minus_mu}
Under the Assumption~\ref{ass:x_s_mu}, the vector $\mu_{i}^{\new} \defeq ( x_i + \wt{\delta}_{x,i} ) ( s_i + \wt{\delta}_{s,i} ) $ satisfies \\
1. $\| \E[ \mu^{-1} ( \mu^{\new} - \mu - \wt{\delta}_{\mu} ) ] \|_2 \leq 10 \epsilon_{\mathrm{mp}} \cdot \epsilon$ and $ \| \E[ \mu^{-1} ( \mu^{\new} - \mu ) ] \|_2 \leq 5 \epsilon$. \\
2. $\Var[ \mu_i^{-1} \mu_i^{\new} ] \leq 50 \epsilon^2 / k$ for all $i$. \\
3. $\| \mu^{-1} ( \mu^{\new} - \mu ) \|_{\infty} \leq \frac{0.021}{\log{n}}$.
\end{lemma}

\begin{claim}[Part 1 of Lemma~\ref{lem:bounding_mu_new_minus_mu}]
\begin{align*}
\| \E[ \mu^{-1} ( \mu^{\new} - \mu - \wt{\delta}_{\mu} ) ] \|_2 \leq 10 \epsilon_{\mathrm{mp}} \cdot \epsilon,
\text{~and~} \| \E[ \mu^{-1} ( \mu^{\new} - \mu ) ] \|_2 \leq 5 \epsilon.
\end{align*}
\end{claim}

\begin{proof}
\begin{align*}
\mu^{\new} =  ( x + \wt{\delta}_x ) ( s + \wt{\delta}_s ) 
= \mu + x \wt{\delta}_s + s \wt{\delta}_x + \wt{\delta}_x \wt{\delta}_s 
= \mu + \underbrace{ \ov{x} \wt{\delta}_s + \ov{s} \wt{\delta}_x }_{ \wt{\delta}_{\mu} } + \underbrace{ ( x - \ov{x} ) \wt{\delta}_s + ( s - \ov{s} ) \wt{\delta}_x + \wt{\delta}_x \wt{\delta}_s }_{ \epsilon_{\mu} }.
\end{align*}
Taking the expectation on both sides, we have
$$ \E[\mu^{\new} - \mu - \wt{\delta}_{\mu} ] = ( x - \ov{x} ) \E[ \wt{\delta}_s ] + (s-\ov{s}) \E[ \wt{\delta}_x ] + \E[ \wt{\delta}_x \wt{\delta}_s ] .$$

Hence, we have that
\begin{align}\label{eq:bounding_mu_new_minus_mu_l2_norm_1}
 & ~ \| \mu^{-1}  \E [ \mu^{\new}  - \mu - \wt{\delta}_{\mu} ]  \|_2 \notag \\
\leq & ~ \| \mu^{-1} (x - \ov{x} ) s \cdot s^{-1} \E[ \wt{\delta}_s ] \|_2 + \| \mu^{-1} ( s - \ov{s} ) x \cdot x^{-1} \E[ \wt{\delta}_x ] \|_2 + \| \mu^{-1} \E[ \wt{\delta}_x \wt{\delta}_s ] \|_2 \notag \\
\leq & ~ \| \mu^{-1} ( x - \ov{x} ) s \|_{\infty} \cdot \| s^{-1} \E [ \wt{ \delta_s } ] \|_2 + \| \mu^{-1} ( s  - \ov{s} ) x \|_{\infty} \cdot \| x^{-1} \E [ \wt{ \delta_x } ] \|_2 + \| \mu^{-1} \E[ \wt{\delta}_x \wt{\delta}_s ] \|_2 \notag \\
\leq & ~ \epsilon_{\mathrm{mp}} \cdot \| s^{-1} \E[ \wt{\delta}_s ] \|_2 + \epsilon_{\mathrm{mp}} \cdot \| x^{-1} \E[ \wt{\delta}_x ] \|_2 + \| \mu^{-1} \E[ \wt{\delta}_x \wt{\delta}_s ] \|_2 \notag \\
\leq & ~ 4 \epsilon_{\mathrm{mp}} \cdot \epsilon + \| \mu^{-1} \E [ \wt{\delta}_x \wt{\delta}_s ] \|_2,
\end{align}
where we used the triangle inequality in the first step, $\| a b\|_2 \leq \| a \|_{\infty} \cdot \| b \|_2$ in the second step, $\| \mu^{-1} ( x - \ov{x} ) s \|_{\infty} \leq \epsilon_{\mathrm{mp}}$ and $\| \mu^{-1} ( s - \ov{s} ) x \|_{\infty} \leq \epsilon_{\mathrm{mp}}$ (since $\ov{x} \approx_{\epsilon_{\mathrm{mp}}} x$, $\ov{s} \approx_{\epsilon_{\mathrm{mp}}} s$) in the third step, and $\| \E [ s^{-1} \wt{\delta}_s ] \|_2 \leq 2 \epsilon$ and $\| \E [ x^{-1} \wt{\delta}_x ] \|_2 \leq 2 \epsilon$ (Part 1 of Lemma~\ref{lem:stochastic_step}) at the end. 

To bound the last term, using $\E[\wt{\delta}_s] = \delta_s$ and $\E[\wt{\delta}_x ] = \delta_x$, we note that
\begin{align*}
\E[ \wt{\delta}_{x,i} \wt{\delta}_{s,i} ] = \delta_{x,i} \delta_{s,i} + \E [ ( \wt{\delta}_{x,i} - \delta_{x,i} ) ( \wt{\delta}_{s,i} - \delta_{s,i} ) ].
\end{align*}
Hence, we have
\begin{align}\label{eq:bounding_mu_new_minus_mu_l2_norm_2}
\| \mu^{-1} \E [ \wt{\delta}_x \wt{\delta}_s ] \|_2 
\leq & ~ \| \mu^{-1} \delta_x \delta_s \|_2 + \left( \sum_{i=1}^n \left( \E \left[ x_i^{-1} ( \wt{\delta}_{x,i} - \delta_{x,i} ) \cdot s_i^{-1} ( \wt{\delta}_{s,i} - \delta_{s,i} ) \right] \right)^2 \right)^{1/2} \notag \\
\leq & ~ 4 \epsilon^2 + \frac{1}{2} \left( \sum_{i=1}^n \left( \Var[  x_i^{-1} \wt{\delta}_{x,i} ]  + \Var[ s_i^{-1} \wt{\delta}_{s,i}  ] \right)^2 \right)^{1/2} \notag \\
\leq & ~ 4 \epsilon^2 + \frac{1}{2} \left( \sum_{i=1}^n 2 ( \Var[ x_i^{-1} \wt{\delta}_{x,i} ] )^2 + 2 ( \Var[ s_i^{-1} \wt{\delta}_{s,i} ] )^2 \right)^{1/2} \notag \\
\leq & ~ 4 \epsilon^2 + 2 \sqrt{ n \cdot \epsilon^4 / k^2  } 
\leq  4 \epsilon^2 + 2 \epsilon  \cdot \epsilon_{\mathrm{mp}} \leq 6 \epsilon \cdot \epsilon_{\mathrm{mp}},
\end{align}
where we used $\| \mu^{-1} \delta_x \delta_s \|_2 \leq \| x^{-1} \delta_x \|_2 \cdot \| s^{-1} \delta_s \|_2 \leq 4 \epsilon^2$ (Part 1 of Lemma~\ref{lem:stochastic_step}) and $2ab \leq a^2 + b^2$ in the second step, $(a+b)^2 \leq 2 a^2 + 2 b^2$ in the third step, $\Var[ x_i^{-1} \wt{\delta}_{x,i} ] \leq 2 \epsilon^2 / k$ and $\Var[ s_i^{-1} \wt{\delta}_{s,i} ] \leq 2 \epsilon^2 / k$ (Part 2 of Lemma~\ref{lem:stochastic_step}) in the fourth step, and $k \geq \frac{\epsilon \sqrt{n}}{\epsilon_{\mathrm{mp}}}$ at the end.

Combining  \eqref{eq:bounding_mu_new_minus_mu_l2_norm_1} and  \eqref{eq:bounding_mu_new_minus_mu_l2_norm_2}, we have that
\begin{align*}
\| \mu^{-1} ( \E [ \mu^{\new} - \mu - \wt{\delta}_{\mu} ] ) \|_2 \leq 4 \epsilon_{\mathrm{mp}} \cdot \epsilon + \| \mu^{-1} \E[ \wt{\delta}_x \wt{\delta}_s ] \|_2 \leq 10 \epsilon_{\mathrm{mp}} \cdot \epsilon.
\end{align*}
where we used $\epsilon \leq \epsilon_{\mathrm{mp}}$.

From Part 1 of Lemma~\ref{lem:stochastic_step}, we know that $\| \mu^{-1} \E[ \wt{\delta}_{\mu} ] \|_2 \leq 4 \epsilon$. Thus using triangle inequality, we know
\begin{align*}
\| \mu^{-1} ( \E [ \mu^{\new} - \mu ] ) \|_2 \leq 10 \epsilon_{\mathrm{mp}} \cdot \epsilon + 4 \epsilon \leq 5 \epsilon.
\end{align*}
\end{proof}

\begin{claim}[Part 2 of Lemma~\ref{lem:bounding_mu_new_minus_mu}]
$\Var[ \mu_i^{-1} \mu_i^{\new} ] \leq 50 \epsilon^2 / k$ for all $i$.
\end{claim}
\begin{proof}
Recall that
\begin{align*}
\mu^{\new} = \mu + \wt{\delta}_{\mu} + ( x - \ov{x} ) \wt{\delta}_s + ( s - \ov{s} ) \wt{\delta}_x + \wt{\delta}_x \wt{\delta}_s.
\end{align*}

We can upper bound the variance of $\mu_i^{-1} \mu_i^{\new}$,
\begin{align*}
 \Var[ \mu_i^{-1} \mu_i^{\new} ] 
\leq & ~ 4 \Var [ \mu_i^{-1} \wt{\delta}_{\mu,i} ] + 4 \Var [ \mu_i^{-1} (x_i -\ov{x}_i) \wt{\delta}_{s,i} ] + 4 \Var [ \mu_i^{-1} ( s_i - \ov{s}_i ) \wt{\delta}_{x,i} ] + 4 \Var [ \mu_i^{-1} \wt{\delta}_{x,i} \wt{\delta}_{s,i} ] \\
\leq & ~ 32\frac{\epsilon^2}{k} + 4\frac{\epsilon^2}{k} + 4\frac{\epsilon^2}{k} + \Var[ \mu_i^{-1} \wt{\delta}_{x,i} \wt{\delta}_{s,i} ] \\
= & ~ 40 \frac{\epsilon^2}{k} + \Var[ x_i^{-1} \wt{\delta}_{x,i} \cdot s_i^{-1} \wt{\delta}_{s,i} ] \\
\leq & ~ 40 \frac{ \epsilon^2 }{ k } + 2 \Sup[(x_i^{-1} \wt{\delta}_{x,i} )^2] \cdot \Var[s_i^{-1} \wt{\delta}_{s,i}] + 2 \Sup[(s_i^{-1} \wt{\delta}_{s,i} )^2] \cdot \Var[ x_i^{-1} \wt{\delta}_{x,i} ]  \\
\leq & ~ 40 \frac{ \epsilon^2 }{ k } + 2 \cdot (\frac{0.02}{\log{n}})^2 \cdot \frac{ \epsilon^2 }{ k } + 2 \cdot (\frac{0.02}{\log{n}})^2 \cdot \frac{\epsilon^2}{k} 
\leq 50 \frac{\epsilon^2}{k}.
\end{align*}
where the second step follows by the inequality $\Var[ \mu_i^{-1} \wt{\delta}_{\mu,i} ] \leq 8 \epsilon^2 / k$ (Part 2 of Lemma~\ref{lem:stochastic_step}),
\begin{align*}
\Var[ \mu_i^{-1} (x_i - \ov{x}_i) \wt{\delta}_{s,i} ]  
= \Var[ x_i^{-1} (x_i - \ov{x}_i) s_i^{-1} \wt{\delta}_{s,i} ] 
\leq 2 \epsilon_{\mathrm{mp}}^2 \Var[ s_i^{-1} \wt{\delta}_{s,i} ]  
\leq \epsilon^2 / k. 
\end{align*}
and a similar inequality $\Var[ \mu_i^{-1} (s_i - \ov{s}_i) \wt{\delta}_{x,i} ] \leq \epsilon^2 / k $, the third step follows by the definition $\mu = xs$, the fourth step follows by the inequality $\Var[xy]\leq 2 \Sup[x^2] \Var[y] + 2 \Sup[y^2] \Var[x]$ (Lemma~\ref{lem:var_xy}) with $\Sup$ denoting the deterministic maximum of the random variable, the fifth step follows by the inequalities $\Var[ s_i^{-1} \wt{\delta}_{s,i} ] \leq 2 \epsilon^2 / k$ and $\Var[ x_i^{-1} \wt{\delta}_{x,i} ] \leq 2 \epsilon^2 / k $ (Part 2 of Lemma~\ref{lem:stochastic_step}).
\end{proof}

\begin{claim}[Part 3 of Lemma~\ref{lem:bounding_mu_new_minus_mu}]
$\| \mu^{-1} ( \mu^{\new} - \mu ) \|_{\infty} \leq \frac{0.021}{\log{n}}$.
\end{claim}
\begin{proof}
We again note that
\begin{align*}
\mu^{\new} = \mu + \wt{\delta}_{\mu} + ( x - \ov{x} ) \wt{\delta}_s + ( s - \ov{s} ) \wt{\delta}_x + \wt{\delta}_x \wt{\delta}_s.
\end{align*}
Hence, we have
\begin{align*}
 & ~ | \mu_i^{-1} ( \mu_i^{\new} - \mu_i - \wt{\delta}_{\mu,i} ) | \\
\leq & ~ | ( x - \ov{x} )_i \mu_i^{-1} \wt{\delta}_{s,i} | + | (s - \ov{s})_i \mu_i^{-1} \wt{\delta}_{x,i} | + | \mu_i^{-1} \wt{\delta}_{x,i} \wt{\delta}_{s,i} | \\
= & ~ | ( x - \ov{x} )_i x_i^{-1} | \cdot | s_i^{-1} \wt{\delta}_{s,i} | + | (s - \ov{s})_i s_i^{-1} | \cdot | x_i^{-1} \wt{\delta}_{x,i} | + | x_i^{-1} \wt{\delta}_{x,i} | \cdot | s_i^{-1} \wt{\delta}_{s,i} | \\
\leq & ~ \epsilon_{\mathrm{mp}} | s_i^{-1} \wt{\delta}_{s,i} | + \epsilon_{\mathrm{mp}} | x_i^{-1} \wt{\delta}_{x,i} | + | s_i^{-1} \wt{\delta}_{s,i} | |x_i^{-1} \wt{\delta}_{x,i} | \\
\leq & ~ \epsilon_{\mathrm{mp}} \cdot \frac{0.2}{\log{n}} + \epsilon_{\mathrm{mp}} \cdot \frac{0.02}{\log{n}} + (\frac{0.02}{\log{n}})^2 \leq \frac{1}{1000 \log{n}},
\end{align*}
where the first step follows by the triangle inequality, the second step follows by the definition $\mu_i = x_i s_i$, the third step follows by the invariants $x \approx_{\epsilon_{\mathrm{mp}}} \ov{x}$ and $s \approx_{\epsilon_{\mathrm{mp}}} \ov{s}$,  the fifth step follows by the inequalities $ | s_i^{-1} \wt{\delta}_{s,i} | \leq \frac{0.02}{\log{n}}$ and $| x_i^{-1} \wt{\delta}_{x,i} | \leq \frac{0.02}{\log{n}}$ (Part 3 of Lemma~\ref{lem:stochastic_step}).

Since we know that $| \mu_i^{-1} \wt{\delta}_{\mu,i} | \leq \frac{0.02}{\log{n}}$ (Part 3 of Lemma~\ref{lem:stochastic_step}), we have
\begin{align*}
| \mu_i^{-1} ( \mu_i^{\new} - \mu_i ) | \leq \frac{1}{1000 \log{n}} + \frac{0.02}{\log{n}} \leq \frac{0.021}{\log{n}} .
\end{align*}
\end{proof}


\subsection{Stochastic central path}
Now, we are ready to prove $x_i s_i \approx_{0.1} t$ during the whole algorithm. As explained in the proof outline (see Section~\ref{sec:stochastic_outline}), we will prove this bound by analyzing the  potential $\Phi_{\lambda}(\mu/t-1)$ where  $\Phi_{\lambda}(r) = \sum_{i=1}^n \cosh (\lambda r_i)$.

First, we give some basic properties of $\Phi_{\lambda}$.
\begin{lemma}[Basic properties of potential function]\label{lem:potential_function_Phi_cosh}
Let $\Phi_{\lambda}(r) = \sum_{i=1}^n \cosh (\lambda r_i)$ for some $\lambda > 0$. For any vector $r \in \R^n$,\\
1. For any vector $\| v \|_{\infty} \leq 1/\lambda$, we have that
\begin{align*}
\Phi_{\lambda}( r + v ) \leq \Phi_{\lambda} (r) + \langle \nabla \Phi_{\lambda} (r), v \rangle + 2 \| v \|_{\nabla^2 \Phi_{\lambda}(r)}^2.
\end{align*}
2.
$
\| \nabla \Phi_{\lambda}(r) \|_2 \geq \frac{\lambda}{ \sqrt{n} } ( \Phi_{\lambda}(r) - n ) .
$\\
3.
$
\left( \sum_{i=1}^n \lambda^2 \cosh^2 (\lambda r_i) \right)^{1/2} \leq \lambda\sqrt{n} + \| \nabla \Phi_{\lambda} (r) \|_2.
$
\end{lemma}
\begin{proof}
For each $i\in [n]$, we use $r_i$ to denote the $i$-th coordinate of vector $r$.

{\bf Proof of Part 1.} 
Using mean-value forms of Taylor's theorem, we have that
\begin{align*}
\cosh( \lambda ( r_i + v_i ) ) = \cosh( \lambda r_i ) + \lambda \sinh ( \lambda r_i ) v_i + \frac{ \lambda^2 }{ 2 } \cosh (\zeta_i) v_i^2,
\end{align*}
where $\zeta_i$ is between $\lambda r_i$ and $\lambda (r_i + v_i)$. By definition of $\cosh$ and the assumption that $\| v \|_{\infty} \leq \frac{1}{ 2 \lambda }$ , we have that
\begin{align*}
\cosh( \zeta_i ) 
=  \frac{1}{2} \exp( \zeta_i ) + \frac{1}{2} \exp( - \zeta_i ) 
\leq  \exp(1) \cdot \frac{1}{2} ( \exp( \lambda r_i ) + \exp( - \lambda r_i ) ) 
\leq  3 \cosh (\lambda r_i).
\end{align*}
Hence, we have
\begin{align*}
\cosh ( \lambda ( r_i + v_i ) ) \leq \cosh ( \lambda r_i ) + \lambda \sinh( \lambda r_i ) v_i + 2 \lambda^2 \cosh ( \lambda r_i ) v_i^2.
\end{align*}
Summing over all the coordinates gives
\begin{align*}
& ~ \sum_{i=1}^n  \cosh( \lambda (r_i+v_i)  ) \leq \sum_{i=1}^n \left[\cosh ( \lambda r_i ) + 2 \lambda \sinh( \lambda r_i ) v_i + \lambda^2 \cosh ( \lambda r_i ) v_i^2 \right] \\
\implies & ~ \Phi_{\lambda}(r+v) \leq \Phi_{\lambda}(r) + \langle \nabla \Phi_{\lambda}(r), v \rangle + 2 \| v \|_{\nabla^2 \Phi_{\lambda}(r)}^2.
\end{align*}

{\bf Proof of Part 2.}
Since $\Phi_{\lambda}(r) = \sum_{i=1}^n \cosh(\lambda r_i)$, then $$\nabla \Phi_{\lambda}(r) = \begin{bmatrix} \lambda \sinh(\lambda r_1) & \lambda \sinh(\lambda r_2) & \cdots & \lambda \sinh(\lambda r_n) \end{bmatrix}^\top.$$
Thus, we can lower bound $\| \nabla \Phi_{\lambda}(r) \|_2$ in the following way,
\begin{align*}
\| \nabla \Phi_{\lambda}(r) \|_2 
= & ~ \left( \sum_{i=1}^n \lambda^2 \sinh^2 ( \lambda r_i ) \right)^{1/2} \\
= & ~ \left( \sum_{i=1}^n \lambda^2 ( \cosh^2 ( \lambda r_i ) - 1) \right)^{1/2} & \text{~by~} \cosh^2(y) - \sinh^2 (y) = 1, \forall y \\
\geq & ~ \frac{\lambda}{ \sqrt{n} } \sum_{i=1}^n \sqrt{ \cosh^2 ( \lambda r_i ) - 1 } & \text{~by~}\| \cdot \|_2 \geq \frac{1}{\sqrt{n}} \|  \cdot \|_1 \\
\geq & ~ \frac{\lambda}{ \sqrt{n} } \sum_{i=1}^n (\cosh( \lambda r_i) - 1) & \text{~by~$\cosh( \lambda r_i) \geq 1$} \\
= & ~ \frac{\lambda}{\sqrt{n}} ( \Phi_{\lambda}(r) - n ). & \text{~by~def~of~}\Phi(r)
\end{align*}

{\bf Proof of Part 3.}

\begin{align*}
\left( \sum_{i=1}^n \lambda^2 \cosh^2 (\lambda r_i) \right)^{1/2} 
= & ~ \left( \sum_{i=1}^n \lambda^2 + \lambda^2 \sinh^2 (\lambda r_i) \right)^{1/2} & \text{~by~} \cosh^2 (y) - \sinh^2(y) = 1, \forall y \\
\leq & ~ ( n \lambda^2 )^{1/2} + \left( \sum_{i=1}^n \lambda^2 \sinh^2 (\lambda r_i) \right)^{1/2} \\
= & ~ \lambda \sqrt{n} + \| \nabla \Phi_{\lambda}(r) \|_2.
\end{align*}
\end{proof}

The following lemma shows that the potential $\Phi$ is decreasing in expectation when $\Phi$ is large.
\begin{lemma} \label{lem:potential_martingale}
Under the Assumption~\ref{ass:x_s_mu}, we have
\begin{align*}
\E \left[ \Phi_{\lambda} \left( \frac{ \mu^{\new} }{ t^{\new} } - 1 \right) \right] \leq \Phi_{\lambda} \left( \frac{ \mu }{ t } - 1 \right) - \frac{\lambda \epsilon}{15\sqrt{n}} \left( \Phi_{\lambda} \left( \frac{ \mu }{ t } - 1 \right) - 10 n \right).
\end{align*}
\end{lemma}
\begin{proof}
Let $\epsilon_{\mu} = \mu^{\new} - \mu - \wt{\delta}_{\mu}$. From the definition, we have
\begin{align*}
 \mu^{\new} - t^{\new} 
= \mu + \wt{\delta}_{\mu} + \epsilon_{\mu} - t^{\new},
\end{align*}
which implies
\begin{align}\label{eq:rewrite_mu_t_new}
\frac{ \mu^{\new} }{ t^{\new} } - 1 = & ~ \frac{\mu}{t^{\new}} + \frac{1}{t^{\new} } ( \wt{\delta}_{\mu} + \epsilon_{\mu} )  -  1 \notag \\
= & ~ \frac{\mu}{t} \frac{t}{t^{\new}} + \frac{1}{t^{\new} } ( \wt{\delta}_{\mu} + \epsilon_{\mu} )  -  1 \notag \\
= & ~ \frac{\mu}{t} + \frac{\mu}{t} ( \frac{t}{t^{\new}} - 1 ) + \frac{1}{t^{\new} } ( \wt{\delta}_{\mu} + \epsilon_{\mu} )  -  1 \notag \\
= & ~ \frac{\mu}{t} - 1 + \underbrace{ \frac{\mu}{t} ( \frac{t}{t^{\new}} - 1 ) + \frac{1}{t^{\new} } ( \wt{\delta}_{\mu} + \epsilon_{\mu} ) }_{v}.
\end{align}

To apply Lemma~\ref{lem:potential_function_Phi_cosh} with $r = \mu / t - 1$ and $r+v=\mu^{\new}/t^{\new}-1$, we first compute the expectation of $v$
\begin{align}\label{eq:expectation_v}
\E[v] = & ~ \frac{\mu}{t} ( \frac{t}{t^{\new}} - 1 ) + \frac{1}{t^{\new} } ( \E[ \wt{\delta}_{\mu} ] + \E[ \epsilon_{\mu} ] ) \notag \\
= & ~ \frac{\mu}{t} ( \frac{t}{t^{\new}} - 1 ) + \frac{1}{t^{\new}}  ( \delta_{\mu} + \E[ \epsilon_{\mu} ] ) \notag \\
= & ~ \frac{\mu}{t} ( \frac{t}{t^{\new}} - 1 ) + \frac{1}{t^{\new}}  \left( \left( (\frac{t^{\new}}{t} - 1)   \mu - \frac{\epsilon}{2} t^{\new} \frac{ \nabla \Phi_{\lambda} (\mu/t -1 ) }{ \| \nabla \Phi_{\lambda} (\mu/t -1 ) \|_2 }  \right)    + \E[ \epsilon_{\mu} ] \right) \notag \\
= & ~ - \frac{\epsilon}{2} \frac{ \nabla \Phi_{\lambda} (\mu/t-1) }{ \| \nabla \Phi_{\lambda}(\mu/t-1) \|_2} + \frac{1}{t^{\new}} \E[ \epsilon_{\mu} ],
\end{align}
where the third step follows by the definition of $\delta_{\mu}$.

Next, we bound the $\|v\|_\infty$ as follows
\begin{align*}
\|v\|_{\infty} & \leq \left\| \frac{\mu}{t}(\frac{t}{t^{\new}}-1) \right\|_{\infty}+ \left\| \frac{1}{t^{\new}}(\wt{\delta}_{\mu}+\epsilon_{\mu}) \right\|_{\infty} \leq\frac{\epsilon}{\sqrt{n}}+\frac{\|\mu^{-1}(\mu^{\new}-\mu)\|_{\infty}}{0.9} \\
& \leq \frac{\epsilon}{\sqrt{n}}+\frac{0.021}{0.9 \log n}\leq\frac{1}{\lambda}.
\end{align*}
where we used Part 3 of Lemma~\ref{lem:bounding_mu_new_minus_mu} and $\eps \leq \frac{1}{400 \log{n}}$.

Since $\|v\|_\infty\leq \frac{1}{\lambda}$, we can apply Part 1 of Lemma~\ref{lem:potential_function_Phi_cosh} and get
\begin{align*}
 & ~ \E [ \Phi_{\lambda} ( \mu / t + v - 1 ) ] \\
\leq & ~ \Phi_{\lambda} ( \mu / t - 1 ) +  \langle \nabla \Phi_{\lambda}( \mu / t - 1 ), \E[v] \rangle + 2 \E [ \| v \|_{ \nabla^2 \Phi_{\lambda}(\mu / t + v - 1) }^2 ]  \\
= & ~ \Phi_{\lambda} (\mu / t - 1 ) - \frac{\epsilon}{2} \| \nabla \Phi_{\lambda}( \mu / t - 1 ) \|_2 + \frac{t}{t^{\new}}  \langle \nabla \Phi_{\lambda}( \mu / t - 1 ), \E[ t^{-1} \epsilon_{\mu} ] \rangle + 2 \E [ \| v \|_{ \nabla^2 \Phi_{\lambda}( \mu / t - 1 ) }^2 ] \\
\leq & ~ \Phi_{\lambda} (\mu / t - 1 ) - \frac{\epsilon}{2} \| \nabla \Phi_{\lambda}( \mu / t - 1 ) \|_2 + \frac{t}{t^{\new}} \| \nabla \Phi_{\lambda} (\mu / t - 1) \|_2 \cdot \| \E[t^{-1} \epsilon_{\mu} ] \|_2 + 2 \E [ \| [v] \|_{ \nabla^2 \Phi_{\lambda}( \mu / t - 1 ) }^2 ] \\
\leq & ~ \Phi_{\lambda} (\mu / t - 1 ) - \frac{\epsilon}{2} \| \nabla \Phi_{\lambda}( \mu / t - 1 ) \|_2 + 10 \epsilon_{\mathrm{mp}} \cdot \epsilon \| \nabla \Phi_{\lambda}( \mu / t - 1 ) \|_2 + 2 \E [ \| v \|_{ \nabla^2 \Phi_{\lambda}( \mu / t - 1 ) }^2 ] ,
\end{align*}
where we substituted $\E[v]$ by \eqref{eq:expectation_v} in the second step, we used $\langle a, b \rangle \leq \| a \|_2 \cdot \| b \|_2$ in the third step, and $\| \E[ t^{-1} \epsilon_{\mu} ] \|_2 \leq 10 \epsilon_{\mathrm{mp}} \cdot \epsilon$ (from Part 1 of Lemma~\ref{lem:bounding_mu_new_minus_mu} and $\mu \approx_{0.1} t$) at the end

We still need to bound $ \E[ \| v \|_{ \nabla^2 \Phi_{\lambda}( \mu / t - 1 ) }^2 ]$. Before bounding it, we first bound $\E[v_i^2]$,
\begin{align}\label{eq:bounding_v_i_square}
\E[v_i^2] \leq & ~ 2 \E\left[  \left( \frac{\mu_i }{t} ( \frac{t}{ t^{\new} } - 1 ) \right)^2 \right] + 2 \E \left[ \left( \frac{1}{t^{\new}} (\wt{\delta}_{\mu,i} + \wh{\delta}_{\mu,i} ) \right)^2 \right] \notag \\
\leq & ~ \epsilon^2 / n + 2.5 \E \left[ (  (\mu^{\new}_i - \mu_i) / \mu_i )^2 \right] \notag \\
= & ~ \epsilon^2 / n + 2.5 \Var[ ( \mu_i^{\new} - \mu_i ) / \mu_i ] + 2.5 ( \E[  ( \mu_i^{\new} - \mu_i ) / \mu_i ]  )^2 \notag \\
\leq & ~\epsilon^2 / n + 125 \epsilon^2 / k + 2.5 ( \E[ ( \mu_i^{\new} - \mu_i ) / \mu_i ]  )^2 \notag \\
\leq & ~ 126 \epsilon^2 / k + 3 ( \E[ ( \mu_i^{\new} - \mu_i ) / \mu_i ]  )^2 ,
\end{align}
where we used the definition of $v$ (see \eqref{eq:rewrite_mu_t_new}) in the first step, $\mu \approx_{0.1} t$ and $(t/t^{\new} -1)^2 \leq \epsilon^2/(4n)$ in the second step, $\E[x^2] = \Var[x] + (\E[x])^2$ in the third step, Part 2 of Lemma~\ref{lem:bounding_mu_new_minus_mu} in the fourth step, and $n \geq k$ at the end.

Now, we are ready to bound $\E[ \| v \|_{ \nabla^2 \Phi_{\lambda}( \mu / t - 1 ) }^2 ] $
\begin{align*}
 & ~ \E[ \| v \|_{ \nabla^2 \Phi_{\lambda}( \mu / t - 1 ) }^2 ] \\
= & ~ \lambda^2 \sum_{i=1}^n \E[ \Phi_{\lambda}(\mu / t -1 )_i v_i^2 ] \\
\leq & ~ \lambda^2 \sum_{i=1}^n \Phi_{\lambda}(\mu / t - 1 )_i \cdot ( 126 \epsilon^2 / k +3 ( \E[ ( \mu_i^{\new} - \mu_i ) /  \mu_i  ]  )^2 )\\
= & ~ 126 \frac{\lambda^2 \epsilon^2}{k} \Phi_{\lambda} ( \mu / t - 1 ) + 3 \lambda^2 \sum_{i=1}^n \Phi_{\lambda} ( \mu / t - 1 )_i \cdot ( \E[ ( \mu_i^{\new} - \mu_i ) /  \mu_i  ]  )^2  \\
\leq & ~ 126 \frac{\lambda^2 \epsilon^2}{k} \Phi_{\lambda} ( \mu / t - 1 ) + 
3 \lambda \left( \sum_{i=1}^n \lambda^2 \Phi_{\lambda} ( \mu / t - 1 )_i^2  \right)^{1/2}
\cdot \| \E[ \mu^{-1} ( \mu^{\new} - \mu ) ] \|_4^2 \\
\leq & ~ 126 \frac{\lambda^2 \epsilon^2}{k} \Phi_{\lambda} ( \mu / t - 1 ) + 
3 \lambda \left( \lambda \sqrt{n} + \| \nabla \Phi_{\lambda}(\mu /t - 1) \|_2 \right)
\cdot 
 (5\eps)^2,
\end{align*} 
where the first step follows from the fact $\Phi_{\lambda}(x)_i = \cosh( \lambda x_i )$, 
the second step follows from  \eqref{eq:bounding_v_i_square},
the fourth step follows from Cauchy-Schwarz inequality,
the fifth step follows from Part 3 of Lemma~\ref{lem:potential_function_Phi_cosh} and the fact that
$\| \E[ \mu^{-1} ( \mu^{\new} - \mu ) ] \|_4^2 \leq \| \E[ \mu^{-1} ( \mu^{\new} - \mu ) ] \|_2^2 \leq (5\eps)^2$ (Lemma~\ref{lem:bounding_mu_new_minus_mu}).

Then,
\begin{align*}
 & ~ \E[ \Phi_{\lambda} ( \mu / t + v - 1 )  ] \\
\leq & ~ \Phi_{\lambda} ( \mu / t - 1 ) - ( \frac{\epsilon}{2} - 10 \epsilon_{\mathrm{mp}} \cdot \epsilon ) \| \nabla \Phi_{\lambda}(\mu /t -1) \|_2  + 252 \frac{\lambda^2 \epsilon^2}{k} \Phi_{\lambda}(\mu / t - 1) \\
& ~ + 150 \lambda^2 \eps^2  \sqrt{n} + 150 \lambda  \eps^2 \| \Phi_{\lambda} ( \mu / t - 1 ) \|_2\\
\leq & ~ \Phi_{\lambda} ( \mu / t - 1 ) - \frac{\epsilon}{3} \| \nabla \Phi_{\lambda}(\mu /t -1) \|_2  + 252 \frac{\lambda^2 \epsilon^2}{k} \Phi_{\lambda}(\mu / t - 1)  + 150 \lambda^2 \eps^2  \sqrt{n}\\
\leq & ~ \Phi_{\lambda} ( \mu / t - 1 ) - \frac{\lambda \epsilon}{3\sqrt{n}}  (\Phi_{\lambda} ( \mu / t - 1 ) - n) + 252 \frac{\lambda^2 \epsilon^2}{k} \Phi_{\lambda}(\mu / t - 1)  + 150 \lambda^2 \eps^2  \sqrt{n}\\
\leq & ~ \Phi_{\lambda} ( \mu / t - 1 ) - \frac{\lambda \epsilon}{3\sqrt{n}}  (\Phi_{\lambda} ( \mu / t - 1 ) / 5- 2n),
\end{align*}
where the second step follows from the inequalities $1000 \lambda \eps \leq 1$ and $1000 \eps_{\mathrm{mp}} \leq 1$,
the third step follows from Part 2 of Lemma~\ref{lem:potential_function_Phi_cosh},
and the last step follows from the inequalities $1000 \lambda \eps_{\mathrm{mp}} \leq \log{n}$ and $k \geq \frac{\sqrt{n} \epsilon \log{n}}{\epsilon_{\mathrm{mp}}}$.
\end{proof}

As a corollary, we have the following:
\begin{lemma} \label{lem:asump}
During the $\textsc{Main}$ algorithm, Assumption~\ref{ass:x_s_mu}
is always satisfied. Furthermore, the $\textsc{ClassicalStep}$ happens
with probability $O(\frac{1}{n^2})$ each step.
\end{lemma}
\begin{proof}
The second and the fourth assumptions simply follow from the choice of $\epsilon_{\mathrm{mp}}$ and $k$.

Let $\Phi^{(k)}$ be the potential at the $k$-th iteration of the
$\textsc{Main}$. The $\textsc{ClassicalStep}$ ensures that $\Phi^{(k)}\leq n^{3}$
at the end of each iteration. By the definition of $\Phi$ and the
choice of $\lambda$ in $\textsc{Main}$, we have that
\begin{align*}
\left\| \frac{xs}{t}-1 \right\|_{\infty} \leq \frac{\ln(2n^{3})}{\lambda}\leq0.1.
\end{align*}
This proves the first assumption $xs \approx_{0.1} t$ with $t>0$.

For the third assumption, we note that
\begin{align*}
\|\delta_{\mu}\|_{2} & =\left\Vert \left(\frac{t^{\new}}{t}-1\right)xs-\frac{\epsilon}{2}\cdot t^{\new}\cdot\frac{\nabla\Phi_{\lambda}(\mu/t-1)}{\|\nabla\Phi_{\lambda}(\mu/t-1)\|_{2}}\right\Vert _{2}\\
 & \leq \left| \frac{t^{\new}}{t}-1 \right| \|xs\|_{2}+\frac{\eps}{2}t^{\new}\\
 & \leq\frac{\epsilon}{3\sqrt{n}}\cdot1.1\sqrt{n} t +1.01\cdot\frac{\eps}{2}t\leq\eps t,
\end{align*}
where we used $xs \approx_{0.1} t$ and the formula of $t^{\new}$. Hence, we proved all assumptions in Assumption~\ref{ass:x_s_mu}.

Now, we bound the probability that $\textsc{ClassicalStep}$ happens. 
In the beginning of the \textsc{Main}, Lemma \ref{lem:feasible_LP} is used to modify the linear program with parameter $\min(\frac{\delta}{2},\frac{1}{\lambda})$. Hence, the initial point $x$ and $s$ satisfies $x s \approx_{1/\lambda} 1$. Therefore, we have $\Phi^{(0)} \leq 10 n$.
Lemma~\ref{lem:potential_martingale} shows $\E [ \Phi^{(k+1)} ] \leq (1-\frac{\lambda\epsilon}{15\sqrt{n}})\E [ \Phi^{(k)} ] +\frac{\lambda\epsilon}{15\sqrt{n}}10 n$.
By induction, we have that $\E [ \Phi^{(k)} ] \leq 10 n$ for all $k$. Since
the potential is positive, Markov inequality shows that for any $k$,
$\Phi^{(k)}\geq n^{3}$ with probability at most $O(\frac{1}{n^2})$.
\end{proof}

\subsection{Analysis of cost per iteration}
To apply the data structure for projection maintenance (Theorem~\ref{thm:maintain_projection}), we need to first prove the input vector $w$ does not change too much for each step.

\begin{lemma} \label{ref:w_movement}
Let $x^{\new} = x + \wt{\delta}_{x}$ and $s^{\new} = s + \wt{\delta}_s$. Let $w = \frac{x}{s}$ and $w^{\new} = \frac{ x^{\new} }{ s^{\new} }$. Then we have
\begin{align*}
\sum_{i=1}^n \left( \E [ \ln w_i^{\new} ] - \ln w_i  \right)^2 \leq 64 \epsilon^2,\quad
 \sum_{i=1}^n \left( \Var [ \ln w_i^{\new} ] \right)^2 \leq 1000 \epsilon^2.
\end{align*} 
\end{lemma}
\begin{proof}

From the definition, we know that
\begin{align*}
\frac{ w_i^{\new} }{w_i} = \frac{1 }{ s_i^{-1} x_i } \frac{ x_i + \wt{\delta}_{x,i} }{ s_i + \wt{\delta}_{s,i} } = \frac{ 1 + x_i^{-1} \wt{\delta}_{x,i} }{ 1 + s_i^{-1} \wt{\delta}_{s,i} }.
\end{align*}

{\bf Part 1.}
For each $i\in [n]$, we have
\begin{align*}
 \E[\ln w_i^{\new}] - \ln w_i 
= & ~ \E\left[ \ln (1 +  x_i^{-1} \wt{\delta}_{x,i})
- \ln (1 + s_i^{-1} \wt{\delta}_{s,i} ) \right] \\
\leq & ~ 2 | \E[ x_i^{-1} \wt{\delta}_{x,i} - s_i^{-1} \wt{\delta}_{s,i} ] | 
& \text{~by~} | s_i^{-1} \wt{\delta}_{s,i} |,| x_i^{-1} \wt{\delta}_{x,i} | \leq 0.2,\text{Lemma~\ref{lem:stochastic_step}} \\
\leq & ~ 2 | \E[x_i^{-1} \wt{\delta}_{x,i}] | + 2 | \E[ s_i^{-1} \wt{\delta}_{s,i} ] |. & \text{~by~triangle~inequality}
\end{align*}
Thus, summing over all the coordinates gives
\begin{align*}
\sum_{i=1}^n \left( \E[\ln w_i^{\new}] - \ln w_i  \right)^2 \leq \sum_{i=1}^n 8 (\E[ x_i^{-1} \wt{\delta}_{x,i}] )^2 + 8 ( \E[s_i^{-1} \wt{\delta}_{s,i}] )^2 \leq 64 \epsilon^2.
\end{align*}
where the first step follows by the triangle inequality, the last step follows by the inequalities $ \|\E[s^{-1} \wt{\delta}_{s}] \|_2^2, \|\E[x^{-1} \wt{\delta}_{x}] \|_2^2 \leq 4 \epsilon^2$ (Part 1 of Lemma~\ref{lem:stochastic_step}).

{\bf Part 2.} For each $i \in [n]$, we have
\begin{align*}
\Var [ w_i^{\new} ]
\leq &
\E \left[ \left( \ln w_i^{\new} - \ln w_i \right)^2 \right] \\
= & ~ \E \left[ \left( \ln \frac{ 1 + x_i^{-1} \wt{\delta}_{x,i} }{ 1 + s_i^{-1} \wt{\delta}_{s,i} } \right)^2 \right] \\
\leq & ~ 2 \E[ ( x_i^{-1} \wt{\delta}_{x,i} - s_i^{-1} \wt{\delta}_{s,i}  )^2 ] \\
\leq & ~ 2 \E[ 2 (x_i^{-1} \wt{\delta}_{x,i})^2 + 2 (s_i^{-1} \wt{\delta}_{s,i})^2 ] \\
= & ~ 4 \E [(x_i^{-1} \wt{\delta}_{x,i})^2] + 4 \E[(s_i^{-1} \wt{\delta}_{s,i})^2] \\
= & ~ 4 \Var[ x_i^{-1} \wt{\delta}_{x,i} ] + 4 ( \E[x_i^{-1} \wt{\delta}_{x,i}])^2 + 4 \Var[ s_i^{-1} \wt{\delta}_{s,i} ] + 4 ( \E[s_i^{-1} \wt{\delta}_{s,i}])^2 \\
\leq & ~ 16 \epsilon^2 / k + 4 ( \E[x_i^{-1} \wt{\delta}_{x,i}])^2 + 4 ( \E[s_i^{-1} \wt{\delta}_{s,i}])^2,
\end{align*}
where we used $\Var[ x_i^{-1} \wt{\delta}_{x,i} ], \Var[ s_i^{-1} \wt{\delta}_{s,i} ] \leq 2 \epsilon^2 / k$ (Part 2 of Lemma~\ref{lem:stochastic_step}) at the end.

Thus summing over all the coordinates
\begin{align*}
\sum_{i=1}^n \left( \Var [ w_i^{\new} ] \right)^2 
\leq & ~  \frac{512 n \epsilon^4}{k^2} + 64 \sum_{i=1}^n \left( (\E[x_i^{-1} \wt{\delta}_{x,i}])^4 + (\E[s_i^{-1} \wt{\delta}_{s,i}])^4 \right) \\
\leq & ~ \frac{512 n \epsilon^4}{k^2} + 2048 \epsilon^4 \leq 1000 \epsilon^2,
\end{align*}
where we used $\| \E [ s^{-1} \wt{\delta}_s ] \|_2^2, \| \E[x^{-1} \wt{\delta}_x ] \|_2^2 \leq 4 \epsilon^2 $ and $k \geq \sqrt{n} \epsilon$ at the end.

\end{proof}

Now, we analyze the cost per iteration in procedure \textsc{Main}. This is a direct application of our projection maintenance result.

\begin{lemma}\label{lem:main_cost}
For $\epsilon \geq \frac{1}{\sqrt{n}}$, each iteration of \textsc{Main} (Algorithm~\ref{alg:main}) takes
$$n^{1+a+o(1)} + \eps \cdot (n^{\omega-1/2 + o(1)} + n^{2-a/2 + o(1)})$$ expected time per iteration in amortized where $0 \leq a \leq \alpha$ controls the batch size in the data structure and $\alpha$ is the dual exponent of matrix multiplication. 
\end{lemma}
\begin{proof}
Lemma \ref{lem:asump} shows that \textsc{ClassicalStep} happens with only $O(1/n^2)$ probability each step. Since the cost of each step only takes $\tilde{O}(n^{2.5})$, the expected cost is only $\tilde{O}(n^{0.5})$.

Lemma~\ref{ref:w_movement} shows that the conditions in Theorem \ref{thm:maintain_projection} holds with the parameter
$C_1 = O(\eps), C_2 = O(\eps), \epsilon_{\mathrm{mp}} = \Theta(1)$.

In the procedure \textsc{StochasticStep}, Theorem \ref{thm:maintain_projection} shows that the amortized time per iteration is mainly dominated by two steps:

1. $\mathrm{mp}.\textsc{Update}(w)$: $O(\eps \cdot (n^{\omega-1/2 + o(1)} + n^{2-a/2 + o(1)})) .$

2. $\mathrm{mp}.\textsc{Query}( \frac{1}{ \sqrt{ \ov{X} \ov{S} } } \wt{\delta}_{\mu} )$: $O(n \cdot \| \wt{\delta}_{\mu} \|_0 + n^{1+a + o(1)} )$.

Combining both running time and using $\E [ \| \wt{\delta}_{\mu} \|_0 ] = O(1+k) = O( \epsilon \sqrt{n} \log^2 n)$ (according to the probability of success in Claim~\ref{claim:prob} and matching Assumption~\ref{ass:x_s_mu}), we have the result.

\end{proof}

\subsection{Main result}

\begin{proof}[Proof of Theorem~\ref{thm:main}]

In the beginning of the $\textsc{Main}$ algorithm, Lemma~\ref{lem:feasible_LP}
is called to modify the linear program. Then, we run
the stochastic central path method on this modified linear program.

When the algorithm stops, we obtain a vector $x$ and $s$ such that
$xs\approx_{0.1}t$ with $t\leq\frac{\delta^{2}}{32n^3}$. Hence, the
duality gap is bounded by $\sum_{i}x_{i}s_{i}\leq(\delta/4n)^{2}$.
Lemma~\ref{lem:feasible_LP} shows how
to obtain an approximate solution of the original linear program with the guarantee
needed using the $x$ and $s$ we just found.

Since $t$ is decreased by $1-\frac{\epsilon}{3\sqrt{n}}$ factor
each iteration, it takes $O(\frac{\sqrt{n}}{\epsilon}\cdot\log(\frac{n}{\delta}))$
iterations in total. In Lemma~\ref{lem:main_cost},
we proved that each iteration takes
\begin{align*}
n^{1+a+o(1)}+\eps\cdot(n^{\omega-1/2+o(1)}+n^{2-a/2+o(1)}).
\end{align*}
and hence the total runtime is
\begin{align*}
O(n^{2.5-a/2+o(1)}+n^{\omega+o(1)} + \frac{n^{1.5+a+o(1)}}{\epsilon}) \cdot \log(\frac{n}{\delta}).
\end{align*}

Since $\epsilon=\Theta(\frac{1}{\log{n}})$, the total runtime is
\begin{align*}
O(n^{2.5-a/2+o(1)}+n^{\omega+o(1)} + n^{1.5+a+o(1)}) \cdot \log(\frac{n}{\delta}).
\end{align*}
Finally, we note that the optimal choice of $a$ is $\min(\frac{2}{3},\alpha)$, which gives the promised runtime.
\end{proof}

Using the same proof, but different choice of the parameters, we can analyze the ultra short step stochastic central path method, where each step involves sampling only polylogarithmic coordinates. As we mentioned before, the runtime is still around $n^\omega$.
\begin{corollary}\label{cor:short_step}
Under the same assumption as Theorem~\ref{thm:main}, if we choose $\epsilon = \Theta(1/\sqrt{n})$ and $a = \min(\frac{1}{3},\alpha)$, the expected time of \textsc{Main} (Algorithm~\ref{alg:main}) is 
\begin{align*}
\left(n^{\omega+o(1)}+n^{2.5-\alpha/2+o(1)}+n^{2+1/3+o(1)}\right)\cdot\log(\frac{n}{\delta}).
\end{align*}
\end{corollary}

\newpage
\section{Projection Maintenance}
\label{sec:projection_maintenance_datastructure}

The goal of this section is to prove the following theorem:

\begin{theorem}[Projection maintenance]\label{thm:maintain_projection}
Given a full rank matrix $A \in \R^{d \times n}$ with $n \geq d$ and a tolerance parameter $0<\epsilon_{\mathrm{mp}}<1/4$.
Given any positive number $a$ such that $a \leq \alpha$ where $\alpha$ is the dual exponent of matrix multiplication.
There is a deterministic data structure (Algorithm~\ref{alg:maintain_projection}) that approximately maintains the projection matrices
$$\sqrt{W} A^\top (A W A^\top)^{-1} A \sqrt{W}$$ for positive diagonal matrices $W$ through the following two operations:
\begin{enumerate}
\item $\textsc{Update}(w)$: Output a vector $\tilde{v}$ such that for all $i$, $$(1-\epsilon_{\mathrm{mp}}) \tilde{v_i} \leq w_i \leq (1+\epsilon_{\mathrm{mp}}) \tilde{v_i}.$$
\item $\textsc{Query}(h)$: Output $\sqrt{\tilde{V}} A^\top (A \tilde{V} A^\top)^{-1} A \sqrt{\tilde{V}} h$ for the $\tilde{v}$ outputted by the last call to $\textsc{Update}$.
\end{enumerate}

The data structure takes $n^2 d^{\omega -2}$ time to initialize and each call of $\textsc{Query}(h)$ takes time 
\begin{align*}
n \cdot \| h \|_0 + n^{1+a+o(1)}.
\end{align*}
Furthermore, if the initial vector $w^{(0)}$ and the (random) update sequence $w^{(1)},\cdots, w^{(T)} \in \R^{n}$ satisfies
\begin{align*}
 \sum_{i=1}^n \left( \E[ \ln w^{(k+1)}_{i} ] - \ln w^{(k)}_{i} \right)^2 \leq C_1^2 \qquad\text{and}\qquad
 \sum_{i=1}^n  
 ( \Var[\ln w^{(k+1)}_{i}] )^2 \leq  C_2^2
\end{align*}
with the expectation and variance is conditional on $w^{(k)}_{i}$ for all $k=0,1,\cdots,T-1$.
Then, the amortized expected time\footnote{If the input is deterministic, so is the output and the runtime.} per call of \textsc{Update}$(w)$ is
$$(C_1  / \epsilon_{\mathrm{mp}} + C_2 / \epsilon_{\mathrm{mp}}^2) \cdot  (n^{\omega-1/2 + o(1)} + n^{2-a/2 + o(1)}) .$$
\end{theorem}

\begin{remark}
For our linear program algorithm, we have $C_1 = O(1/\log n)$, $C_2 = O(1/\log n)$ and $\epsilon_{\mathrm{mp}} = \Theta(1)$. See Lemma~\ref{ref:w_movement}.
\end{remark}

\subsection{Proof outline}\label{sec:projection_outline}

For intuition, we consider the case $C_1 = \Theta(1)$, $C_2 = \Theta(1)$, and $\eps_{\mathrm{mp}} = \Theta(1)$ in this explanation. The correctness of the data structure (Algorithm~\ref{alg:maintain_projection}) directly follows from Woodbury matrix identity. (The update rule in Line~\ref{lin:offlinewoodburry} correctly maintains $M = A^\top ( A V A^\top )^{-1} A$). The amortized time analysis is based on a potential function that measures the distance of the approximate vector $v$ and the target vector $w$. We will show that 
\begin{itemize}
\item The cost to update the projection $M$ is proportional to the decrease of the potential.
\item Each call to query increase the potential by a fixed amount. 
\end{itemize}
Combining both together gives the amortized runtime bound of our data structure.

\begin{algorithm}[H]\caption{\small Projection Maintenance Data Structure}\label{alg:maintain_projection} {\small
\begin{algorithmic}[1]
\State {\bf datastructure} \textsc{MaintainProjection} \Comment{Theorem~\ref{thm:maintain_projection}}
\\
\State {\bf members}
  \State \hspace{4mm} $w \in \R^{n}$  \Comment{Target vector}
  \State \hspace{4mm} $v, \wt{v} \in \R^{n}$   \Comment{Approximate vectors $v \approx_{\epsilon_{\mathrm{mp}}} w$ and $\wt{v} \approx_{\epsilon_{\mathrm{mp}}} w$}
  \State \hspace{4mm} $A \in \R^{d \times n}$
  \State \hspace{4mm} $M \in \R^{n \times n}$  \Comment{Matrix $M = A^\top ( A V A^\top)^{-1} A $}
  \State \hspace{4mm} $\epsilon_{\mathrm{mp}} \in (0,1/4)$ \Comment{Tolerance}
  \State \hspace{4mm} $a \in (0,\alpha]$ \Comment{Batch Size $n^a$ for Update}
\State {\bf end members}
\\
  \Procedure{\textsc{Initialize}}{$A,w,\epsilon_{\mathrm{mp}}, a$} \Comment{Lemma~\ref{lem:maintain_projection_initialization}}
  \State $w \leftarrow w$, $v \leftarrow w$, $\epsilon_{\mathrm{mp}} \leftarrow \epsilon_{\mathrm{mp}}$, $A \leftarrow A$, $a \leftarrow a$
    \State $M \leftarrow A^\top ( A V A^\top)^{-1} A $
  \EndProcedure
  \\
  \Procedure{\textsc{Update}}{$w^{\new}$} \Comment{Lemma~\ref{lem:maintain_projection_update}}
    \State $y_i \leftarrow \ln w^{\new}_i - \ln v_i$, $\forall i \in [n]$
    \State $r \leftarrow$ the number of indices $i$ such that $|y_i| \geq \epsilon_{\mathrm{mp}} / 2$.
    \If {$r < n^a$}
      \State $v^{\new} \leftarrow v$
      \State $M^{\new} \leftarrow M$
    \Else
      \State Let $\pi : [n] \rightarrow [n]$ be a sorting permutation such that $|y_{\pi(i)}| \geq |y_{\pi(i+1)}|$
      \While{$1.5 \cdot r < n$ and $|y_{\pi(\lceil 1.5 \cdot r  \rceil)}| \geq (1-1/\log n) |y_{\pi(r)}|$}
        \State $r \leftarrow \min(\lceil 1.5 \cdot r  \rceil, n)$
      \EndWhile
      \State $v^{\new}_{\pi(i)} \leftarrow \begin{cases} w^{\new}_{\pi(i)} & i \in \{1,2,\cdots,r\} \\ v_{\pi(i)} & i \in \{r+1, \cdots, n\} \end{cases}$
      \\
      \Comment{Compute $M^{\new} = A^\top ( A V^{\new} A^\top )^{-1} A$ via Woodbury matrix identity}
      \State $\Delta \leftarrow \mathrm{diag}( v^{\new} - v )$ \Comment{$\Delta \in \R^{n \times n}$ and $\| \Delta \|_0 = r$}
      \State Let $S \leftarrow \pi( [r] )$ be the first $r$ indices in the permutation.
      \State Let $M_S\in \R^{n \times r}$ be the $r$ columns from $S$ of $M$.
      \State Let $M_{S,S},\Delta_{S,S} \in \R^{r \times r}$ be the $r$ rows and columns from $S$ of $M$ and $\Delta$.
      \State $M^{\new} \leftarrow M - M_S \cdot ( \Delta^{-1}_{S,S} + M_{S,S} )^{-1} \cdot (M_S)^\top$ \label{lin:offlinewoodburry}
    \EndIf
    \State $w \leftarrow w^{\new}$, $v \leftarrow v^{\new}$, $M \leftarrow M^{\new}$
    \State $\tilde{v}_{i} \leftarrow \begin{cases} v_{i} & \text{if } 
|\ln w_i - \ln v_i| < \epsilon_{\mathrm{mp}}/2\\ 
w_{i} & \text{otherwise} \end{cases}$\label{lin:vtilde}
\State \Return $\tilde{v}$
  \EndProcedure
  \\
  \Procedure{\textsc{Query}}{$h$} \Comment{Lemma~\ref{lem:maintain_projection_query}}
    \State Let $\wt{S}$ be the indices $i$ such that $|\ln w_i - \ln v_i| \geq \epsilon_{\mathrm{mp}}/2$.
    \State \Return $\sqrt{\wt{V}} \cdot ( M \cdot ( \sqrt{\wt{V}} \cdot h ) )  - \sqrt{ \wt{V}} \cdot ( M_{\wt{S}}  \cdot ( ( \wt{\Delta}_{\wt{S},\wt{S}}^{-1} + M_{\wt{S},\wt{S}} )^{-1}  \cdot ( M_{\wt{S}}^\top \sqrt{\wt{V}} h )  ) )$ \label{lin:onlinewoodburry}
  \EndProcedure
  \\
\State {\bf end datastructure}
\end{algorithmic}}
\end{algorithm}

Now, we explain the definition of the potential. Consider the $k$-th round of the algorithm. For all $i \in [n]$, we define $x^{(k)}_{i} = \ln w^{(k)}_{i} - \ln  v^{(k)}_{i}$. Note that $|x^{(k)}_{i}|$ measures the relative distance between $w^{(k)}_i$ and $v^{(k)}_i$. Our algorithm fixes the indices with largest error $x^{(k)}_{i}$. To capture the fact that updating in a larger batch is more efficient, we define the potential as a weighted combination of the error where we put more weight to higher $x^{(k)}_{i}$. Formally, we sort the coordinates of $x^{(k)}$ such that $|x^{(k)}_{i}| \geq |x^{(k)}_{i+1}|$ and define the potential by
\begin{align*}
\Psi_k = \sum_{i = 1}^n g_i \cdot \psi (x^{(k)}_{i}).
\end{align*}
where $g_i$ are positive decreasing numbers to be chosen and $\psi$ is a symmetric ($\psi(x) = \psi(-x)$) positive function that increases on both sides. For intuition, one can think $\psi(x)$ behaves roughly like $|x|$.

Each iteration we update the projection matrix such that the error of $|x_1|, \cdots, |x_r|$ drops from roughly $\epsilon_{\mathrm{mp}}$ to 0. This decreases the potential of $\psi (x^{(k)}_{i})$ by $\Omega(\epsilon_{\mathrm{mp}})$ from $i=1,\cdots,r$. Therefore, the whole potential decreases by $\Omega(\epsilon_{\mathrm{mp}} \sum_{i = 1}^r g_i)$. To make the term $\sum_{i = 1}^r g_i$ proportional to the time to update a rank $r$ part of the projection matrix, we set
\begin{align}\label{eq:g}
g_i =  \begin{cases} n^{-a}, & \text{if~} i < n^a; \\ i^{ \frac{\omega-2}{1-a} - 1 } n^{- \frac{a(\omega-2)}{1-a}}, & \text{otherwise}. \end{cases} 
\end{align}
where $\omega$ is the exponent of matrix multiplication and $a$ is any positive number less than or equals to the dual exponent of matrix multiplication. Lemma \ref{lem:omega_leq_3_minus_a} shows that $g$ is indeed non-increasing and Lemma~\ref{lem:maintain_projection_update} shows that the update time of data-structure is indeed $O(r g_r n^{2+o(1)}) = O(\sum_{i = 1}^r g_i n^{2+o(1)})$ for any $r \geq n^a$.

Each call to $\textsc{Update}$, the expectation of the error vector $x^{(k)}$ moves roughly in an unit $\ell_2$ ball. Therefore, the changes of the potential is roughly upper bounded $(\sum_{i = 1}^n g_i^2)^{1/2} \approx n^{\omega - 5/2}$. Since it takes us $n^{2+o(1)}$ time to decrease the potential by roughly $1$ in the update step, the total time is roughly $n^{\omega-1/2}$.

For the case of stochastic central path, we note that the variance of the vector $x$ is quite small. By choosing a smooth potential function $\psi$ (see (\ref{eq:def_psi})), we can essentially give the same result as if there is no variance.

\subsection{Proof of Theorem~\ref{thm:maintain_projection}}\label{sec:correct}
Now, we give the proof of Theorem~\ref{thm:maintain_projection}. We will defer some simple calculations into later sections.
\begin{proof}[Proof of Theorem~\ref{thm:maintain_projection}]
\ 

{\bf Proof of Correctness.}
The definition of $\tilde{v}$ in Line~\ref{lin:vtilde} ensures that $(1 - \epsilon_{\mathrm{mp}}) \tilde{v_i} \leq w_i \leq (1 + \epsilon_{\mathrm{mp}}) \tilde{v_i}$.

Using the Woodbury matrix identity, one can verify that the update rule in Line~\ref{lin:offlinewoodburry} correctly maintains $M = A^\top ( A V A^\top )^{-1} A$. See the deviation of the formula in Lemma~\ref{lem:maintain_projection_initialization}.
By the same reasoning, the Line~\ref{lin:onlinewoodburry} outputs the vector $\sqrt{\tilde{V}}A^\top ( A \tilde{V} A^\top )^{-1} A \sqrt{\tilde{V}} h$.
This completes the proof of correctness.

{\bf Definition of $x$ and $y$.}
Consider the $k$-th round of the algorithm. For all $i \in [n]$, we define $x^{(k)}_{i}$, $x^{(k+1)}_{i}$ and $y^{(k)}_{i}$ as follows:
\begin{align*}
x^{(k)}_{i} = \ln w^{(k)}_{i} - \ln v^{(k)}_{i}  , y^{(k)}_{i} = \ln w^{(k+1)}_{i} - \ln v^{(k)}_{i} ,  x^{(k+1)}_{i} = \ln w^{(k+1)}_{i} - \ln v^{(k+1)}_{i} .
\end{align*}
Note that the difference between $x^{(k)}_i$ and $y^{(k)}_i$ is that $w$ is changing. The difference between $y^{(k)}_i$ and $x^{(k+1)}_i$ is that $v$ is changing. 

{\bf Assume sorting.} Assume the coordinates of vector $x^{(k)} \in \R^n$ are sorted such that $|x^{(k)}_{i}| \geq |x^{(k)}_{i+1}|$. Let $\tau$ and $\pi$ are permutations such that $|x^{(k+1)}_{\tau(i)}| \geq |x^{(k+1)}_{\tau(i+1)}|$ and $|y^{(k)}_{\pi(i)}| \geq |y^{(k)}_{\pi(i+1)}|$.

{\bf Definition of Potential function.}
Let $g$ be defined in  \eqref{eq:g}. Let $\psi : \R \rightarrow \R$ be defined by
\begin{align}\label{eq:def_psi}
\psi(x) = \begin{cases}
\frac{|x|^2}{\epsilon_{\mathrm{mp}}}, & |x| \in [0,\epsilon_{\mathrm{mp}}/2]\\
\epsilon_{\mathrm{mp}}/2 - \frac{ (\epsilon_{\mathrm{mp}} - |x|)^2 }{\epsilon_{\mathrm{mp}}}, & |x| \in (\epsilon_{\mathrm{mp}}/2,\epsilon_{\mathrm{mp}}] \\
\epsilon_{\mathrm{mp}} /2. & |x| \in (\epsilon_{\mathrm{mp}}, +\infty)
\end{cases}
\end{align}
We define the potential at the $k$-th round by
\begin{align*}
\Psi_k = \sum_{i = 1}^n g_i \cdot \psi (x^{(k)}_{\tau_k(i)}).
\end{align*}
where $\tau_k(i)$ is the permutation such that $|x^{(k)}_{\tau_k(i)}| \geq |x^{(k)}_{\tau_k(i+1)}|$.

{\bf Bounding the potential.}

We can express $\Psi_{k+1} - \Psi_k$ as follows:
\begin{align}\label{eq:potentialchange}
\Psi_{k + 1} - \Psi_{k}
= & ~ \sum_{i=1}^n g_i \cdot \left( \psi (  x^{(k+1)}_{\tau(i)} ) - \psi ( x^{(k)}_{i} ) \right) \notag \\
= & ~ \sum_{i = 1}^n g_i \cdot \underbrace{ \left( \psi ( y^{(k)}_{\pi(i)} ) - \psi ( x^{(k)}_{i} ) \right) }_{w\text{~move}} - \sum_{i = 1}^n g_i \cdot \underbrace{ \left( \psi ( y^{(k)}_{\pi(i)} ) - \psi ( x^{(k+1)}_{\tau(i)} ) \right) }_{v\text{~move}}.
\end{align}
Now, using Lemma~\ref{lem:w_move} and \ref{lem:v_move}, and the fact that $\Psi_0 = 0$ and $\Psi_T \geq 0$, with \eqref{eq:potentialchange}, we get
\begin{align*}
0 
\leq & ~  \Psi_T - \Psi_0  
= \sum_{k = 0}^{T-1} \left( \Psi_{k+1} - \Psi_{k} \right) \\
\leq & ~ \sum_{k = 0}^{T-1}  \left( O(C_1 + C_2 / \epsilon_{\mathrm{mp}}) \cdot \sqrt{\log n} \cdot (n^{-a/2} + n^{\omega-5/2} ) - \Omega (\epsilon_{\mathrm{mp}} r_k g_{r_k} / \log n ) \right) \\
= & ~ T \cdot O(C_1 + C_2 / \epsilon_{\mathrm{mp}}) \cdot \sqrt{\log n} \cdot (n^{-a/2} + n^{\omega-5/2} ) - \sum_{k=1}^T \Omega (\epsilon_{\mathrm{mp}} r_k g_{r_k} / \log n),
\end{align*}
where the third step follows by Lemma~\ref{lem:w_move} and Lemma~\ref{lem:v_move} and $r_k$ is the number of coordinates we update during that iteration.

Therefore, we get, 
\begin{align*}
\sum_{k=1}^T r_k g_{r_k} = O \left( T \cdot (C_1  / \epsilon_{\mathrm{mp}} + C_2 / \epsilon_{\mathrm{mp}}^2)\cdot \log^{3/2} n \cdot (n^{\omega-5/2} + n^{-a/2}) \right).
\end{align*}

{\bf Proof of running time.}
See the Section \ref{sec:time}.
\end{proof}

\subsection{Initialization time, update time, query time}\label{sec:time}

To formalize the amortized runtime proof, we first analyze the initialization time (Lemma~\ref{lem:maintain_projection_initialization}), update time (Lemma~\ref{lem:maintain_projection_update}), and query time (Lemma~\ref{lem:maintain_projection_query}) of our projection maintenance data-structure.

\begin{lemma}[Initialization time]\label{lem:maintain_projection_initialization}
The initialization time of data-structure \textsc{MaintainProjection} (Algorithm~\ref{alg:maintain_projection}) is $O( n^2 d^{\omega - 2} )$.
\end{lemma}
\begin{proof}
Given matrix $A \in \R^{d \times n}$ and diagonal matrix $V \in \R^{n \times n}$, computing $A^\top ( A V A^\top)^{-1} A$ takes $O(n^2 d^{\omega -2})$.
\end{proof}

\begin{lemma}[Update time]\label{lem:maintain_projection_update}
The update time of data-structure \textsc{MaintainProjection} (Algorithm~\ref{alg:maintain_projection}) is $O(r g_{r} n^{2+o(1)})$ where $r$ is the number of indices we updated in $v$. 
\end{lemma}
\begin{proof}

Let $A_S\in \R^{d \times r}$ be the $r$ columns from $S$ of $A$.
From $k$-th query to $(k+1)$-th query, we have
\begin{align*}
 & ~ A^\top (A V^{(k+1)} A^\top)^{-1} A \\
= & ~ A^\top (A (  V^{(k)} + \Delta) A^\top)^{-1} A  \\
= & ~ A^\top \left( (A V^{(k)} A^\top)^{-1} - (A V^{(k)} A^\top)^{-1} A_S ( \Delta^{-1}_{S,S}+ A_S^\top (A V^{(k)} A^\top)^{-1} A_S )^{-1} A_S^\top (A V^{(k)} A^\top)^{-1} \right) A \\
= & ~ A^\top (A V^{(k)} A^\top)^{-1} A - A^\top (A V^{(k)} A^\top)^{-1} A_S ( \Delta^{-1}_{S,S} + A_S^\top (A V^{(k)} A^\top)^{-1} A_S )^{-1} A_S^\top (A V^{(k)} A^\top)^{-1} A \\
= & ~ M^{(k)} - M^{(k)}_S ( \Delta^{-1}_{S,S} + M^{(k)}_{S,S} )^{-1} (M^{(k)}_S)^\top,
\end{align*}
where the second step follows by Woodbury matrix identity and the last step follows by the definition of $M^{(k)} \in \R^{n \times n}$.

Thus the update rule of matrix $M^{(k+1)} \in \R^{n \times n}$ can be written as 
\begin{align*}
M^{(k+1)} = M^{(k)} - M^{(k)}_S ( \Delta^{-1}_{S,S} + (M^{(k)})_{S,S} )^{-1} (M^{(k)}_S)^\top.
\end{align*}

The updates in round $k$ can be splitted into four parts:
\begin{enumerate}
\item Adding two $r \times r$ matrices takes $O(r^2)$ time.
\item Computing the inverse of an $r \times r$ matrix takes $O( r^{\omega + o(1)} )$ time.
\item Computing the matrix multiplication of a $n \times r$ and $r \times n$ matrix takes $O(r g_{r} \cdot n^{2+o(1)})$ time where we used that $r \geq n^a$ (Lemma~\ref{lem:rmw}).
\item Adding two $n \times n$ matrices together takes $O(n^2)$ time.
\end{enumerate}
Hence, the total cost is
\begin{align*}
O( r^2 + r^{\omega + o(1)}  + r g_{r}\cdot n^{2+o(1)} + n^2) = O(r^2 + r^{\omega + o(1)}  + r g_{r}\cdot n^{2+o(1)}) = O(  r g_{r}\cdot n^{2+o(1)}).
\end{align*} 
where we used $r g_{r} \geq 1$ for all $r \geq n^a$ in the first step.
\end{proof}

\begin{lemma}[Query time]\label{lem:maintain_projection_query}
The query time of data-structure \textsc{MaintainProjection} (Algorithm~\ref{alg:maintain_projection}) is $O(n \cdot \| h \|_0 + n^{1+a + o(1)})$.
\end{lemma}
\begin{proof}
Let $\wt{\Delta}$ satisfy $\wt{V} = V + \wt{\Delta}$. Let $\wt{S} \subset [n]$ denote the support of $\wt{\Delta}$ and then $| \wt{S} | \leq n^a$. Let $\wt{r}$ denote $|\wt{S}|$. We abuse the notation here, $\wt{\Delta}$ denotes both $n \times n$ diagonal matrix and a length $n$ vector.

Using Woodbury matrix identity and definition of $M$, the same proof as Update time (Lemma~\ref{lem:maintain_projection_update}) shows 
\begin{align*}
 A^\top ( A \wt{V} A^\top )^{-1} A 
=  M + M_{\wt{S}}  \left( \wt{\Delta}_{\wt{S},\wt{S}}^{-1} + M_{\wt{S},\wt{S}} \right)^{-1}  M_{\wt{S}}^\top,
\end{align*}
where $\wt{\Delta}_{\wt{S}\times \wt{S}}$ has size $\wt{r} \times \wt{r}$, $M_{\wt{S},\wt{S}}$ has size $\wt{r} \times \wt{r}$ and $M_{\wt{S}}$ has size $n \times \wt{r}$.

To compute $\sqrt{\wt{V}} A^\top ( A \wt{V} A^\top )^{-1} A \sqrt{\wt{V}} h$, we just need to compute
\begin{align*}
\sqrt{\wt{V}} M \sqrt{\wt{V}} h + \sqrt{\wt{V}} M_{\wt{S}}  ( \wt{\Delta}_{\wt{S},\wt{S}}^{-1} + M_{\wt{S},\wt{S}} )^{-1} M_{\wt{S}}^\top \sqrt{\wt{V}} h.
\end{align*}
Note the running time of computing the first term of the above equation only takes $O(n \cdot \| h \|_0)$ time.

Next, we analyze the cost of computing the second term of the above equation. It contains several parts:
\begin{enumerate}
\item Computing $\wt{M}_{\wt{S}}^\top \cdot( \sqrt{ \wt{V} } \cdot h ) \in \R^{\wt{r}}$ takes $\wt{r} \| h \|_0$ time.
\item Computing $( \wt{\Delta}_{\wt{S},\wt{S}}^{-1} + M_{\wt{S},\wt{S}} )^{-1} \in \R^{\wt{r} \times \wt{r}}$ that is the inverse of a $\wt{r} \times \wt{r}$ matrix takes $\wt{r}^{\omega + o(1)}$ time.
\item Computing matrix-vector multiplication between $\wt{r} \times \wt{r}$ matrix ($( \wt{\Delta}_{\wt{S},\wt{S}}^{-1} + M_{\wt{S},\wt{S}} )^{-1}$) and $\wt{r} \times 1$ vector ($\wt{M}_{\wt{S}}^\top \sqrt{ \wt{V} } h$) takes $O(\wt{r}^2)$ time.
\item Computing matrix-vector multiplication between $n \times \wt{r}$ matrix ($ M_{\wt{S}} $) and $\wt{r} \times 1$ vector ($( \wt{\Delta}_{\wt{S},\wt{S}}^{-1} + M_{\wt{S},\wt{S}} )^{-1} M_{\wt{S}}^\top \sqrt{\wt{V}} h$) takes $O(n \wt{r})$ time.
\item Computing the entry-wise product of two $n$ vectors takes $O(n)$ time
\end{enumerate}
Thus, overall the running time is
\begin{align*}
O( \wt{r} \| h\|_0 + \wt{r}^{\omega + o(1)} + \wt{r}^2 +  n \wt{r} + n ) = O( \wt{r}^{\omega + o(1)} + n \wt{r} ) = O(n^{a \cdot \omega + o(1)} + n^{1+a}).
\end{align*}

Finally, we note that $\omega \leq 3 -\alpha \leq 3 - a$ (Lemma \ref{lem:omega_leq_3_minus_a}) and hence $a \cdot \omega \leq a(3-a) \leq 1+a$. Therefore, the runtime is $n^{1+a+o(1)}$.
\end{proof}

\subsection{Bounding $w$ move}\label{sec:w_move}
The goal of this section is to prove Lemma~\ref{lem:w_move}.
\begin{lemma}[$w$ move]\label{lem:w_move}
We have 
\begin{align*}
 \sum_{i=1}^n g_i \cdot \E \left[ \psi ( y^{(k)}_{\pi(i)} ) - \psi ( x^{(k)}_{i} ) \right] \leq O(C_1 + C_2 / \epsilon_{\mathrm{mp}}) \cdot \sqrt{\log n} \cdot (n^{-a/2} + n^{\omega-5/2} ).
\end{align*}
\end{lemma}

\begin{proof}

Observe that since the errors $|x^{(k)}_{i}|$ are sorted in descending order, and $\psi(x)$ is symmetric and non-decreasing function for $x\geq 0$, thus $\psi( x^{(k)}_i )$ is also in decreasing order. In addition, note that $g$ is decreasing, we have
\begin{align}\label{eq:sum_g_i_x_pi_i_leq_sum_g_i_x_i}
\sum_{i = 1}^n g_i \psi ( x^{(k)}_{\pi(i)} ) \leq \sum_{i = 1}^n g_i \psi( x^{(k)}_{i} ) .
\end{align}
Hence the first term in \eqref{eq:potentialchange} can be upper bounded as follows:

\begin{align}
 \E \left[ \sum_{i = 1}^n g_i \cdot \left( \psi ( y^{(k)}_{\pi(i)} ) - \psi ( x^{(k)}_{i} ) \right) \right]  
\leq & ~ \E \left[ \sum_{i = 1}^n g_i \cdot \left( \psi ( y^{(k)}_{\pi(i)} ) - \psi (  x^{(k)}_{\pi(i)} ) \right) \right] &  \text{~by~ \eqref{eq:sum_g_i_x_pi_i_leq_sum_g_i_x_i}} \notag \\
= & ~ \sum_{i=1}^n g_i \cdot \E[ \psi ( y^{(k)}_{\pi(i)} ) - \psi (  x^{(k)}_{\pi(i)} ) ]  \notag \\
= & ~ O(C_1 + C_2 / \epsilon_{\mathrm{mp}}) \cdot \sqrt{\log n} \cdot (n^{-a/2} + n^{\omega-5/2} ). & \text{~by~Lemma~\ref{lem:bounding_E_psi_y_minus_psi_x}} \notag
\end{align}
Thus, we complete the proof of $w$ move Lemma.
\end{proof}

It remains to prove the following Lemma,
\begin{lemma}\label{lem:bounding_E_psi_y_minus_psi_x}
\begin{align*}
  \sum_{i=1}^n g_i \cdot \E[ \psi ( y^{(k)}_{\pi(i)} ) - \psi (  x^{(k)}_{\pi(i)} ) ] = O(C_1 + C_2 / \epsilon_{\mathrm{mp}}) \cdot \sqrt{\log n} \cdot (n^{-a/2} + n^{\omega-5/2} ).
\end{align*}
\end{lemma}

\begin{proof}
We separate the term into two:
\begin{align*}
\sum_{i=1}^n g_i \cdot \E[ \psi ( y^{(k)}_{\pi(i)} ) - \psi (  x^{(k)}_{\pi(i)} ) ] 
= \sum_{i=1}^n g_{\pi^{-1}(i)} \cdot \E[ \psi ( y^{(k)}_i ) - \psi (  \E [y^{(k)}_i] ) ] 
+\sum_{i=1}^n g_{\pi^{-1}(i)} \cdot (\psi ( \E [y^{(k)}_i] ) - \psi (  x^{(k)}_i )).
\end{align*}

For the first term, Mean value theorem shows that
\begin{align*}
\psi( y^{(k)}_{i} ) - \psi ( \E [ y^{(k)}_{i} ] )
= & ~ \psi'( \E [ y^{(k)}_{i} ] ) ( y^{(k)}_{i} - \E [ y^{(k)}_{i} ] ) + \frac{1}{2} \psi''(\zeta) (y^{(k)}_{i} - \E [ y^{(k)}_{i} ] )^2 \\
\leq & ~ \psi' ( \E [ y^{(k)}_{i} ] ) (w^{(k+1)}_{i} - \E [ w^{(k+1)}_{i} ] )  + \frac{L_2}{2} \left( w^{(k+1)}_{i} - \E [ w^{(k+1)}_{i} ] \right)^2,
\end{align*}
where $L_2 = \max_x \psi''(x)$.
Let $\gamma_i = \Var[\ln w^{(k+1)}_i]$. Summing over $i$ and taking conditional expectation given $w^{(k)}$ on both sides, we get
\begin{align*}
\sum_{i=1}^n g_{\pi^{-1}(i)}\E[ \psi(y_{i}^{(k)})-\psi(\E [ y_{i}^{(k)} ] )]
\leq & ~ \sum_{i=1}^n g_{\pi^{-1}(i)}\psi'(\E [y_{i}^{(k)}]) \E [ w_{i}^{(k+1)}-\E [ w_{i}^{(k+1)} ] ] +\frac{L_{2}}{2}\sum_{i=1}^n g_{\pi^{-1}(i)}\gamma_{i}\\
= & ~ \frac{L_{2}}{2}\cdot \sum_{i=1}^n g_{\pi^{-1}(i)}\gamma_{i} \\
\leq & ~ \frac{L_{2}}{2}\cdot \|g\|_{2}\cdot \left( \sum_{i=1}^n \gamma_{i}^{2} \right)^{1/2} \\
\leq & ~ \frac{L_{2}}{2}\cdot C_{2}\cdot \|g\|_{2}
\end{align*}

For the second term, we define $\beta_i = \E[ \ln w^{(k+1)}_i] - \ln w^{(k)}_i$. Lipschitz constant of $\psi$ shows that
\begin{align*}
\sum_{i=1}^n g_{\pi^{-1}(i)}(\psi(\E [ y_{i}^{(k)} ] )-\psi(x_{i}^{(k)})) 
\leq & ~ L_{1}\cdot\sum_{i=1}^n g_{\pi^{-1}(i)}|\E [y_{i}^{(k)}]-x_{i}^{(k)}|\\
= & ~ L_{1}\cdot\sum_{i=1}^n g_{\pi^{-1}(i)}|\beta_{i}|\\
\leq & ~ L_{1}\cdot C_{1}\cdot\|g\|_{2}
\end{align*}
where we used that $\sum_{i=1}^n \beta_i^2 \leq C_1^2$.

Now, combining both terms and using that $L_1 = O(1)$, $L_2 = O(1/\epsilon_{\mathrm{mp}})$ (from part 4 of Lemma~\ref{lem:def_psi}) and $\| g\|_2 \leq \sqrt{\log n} \cdot O(n^{-a/2} + n^{\omega-5/2} )$ (from Lemma~\ref{lem:sum_g_i}), we have that
$$  \sum_{i=1}^{n} g_i \cdot \E[ \psi ( y^{(k)}_{\pi(i)} ) - \psi (  x^{(k)}_{\pi(i)} ) ]  \leq O(C_1 + C_2 / \epsilon_{\mathrm{mp}}) \cdot \sqrt{\log n} \cdot (n^{-a/2} + n^{\omega-5/2} ).$$
\end{proof}

\begin{lemma}\label{lem:sum_g_i}
\begin{align*}
\left( \sum_{i=1}^n g_i^2 \right)^{1/2} \leq \sqrt{\log n} \cdot O( n^{-a/2} + n^{\omega-5/2} ) .
\end{align*}
\end{lemma}
\begin{proof}
Since function $g$ behaves differently when $i\leq n^a$ and $i > n^a$. We split the sum into two parts.

For the first part, we have
\begin{align*}
\sum_{i=1}^{n^a} g_i^2  = \sum_{i=1}^{n^a} n^{-2a} = n^{-a}.
\end{align*}

For the second part, we have
\begin{align*}
\sum_{i=n^a}^{n} g_i^2 = \sum_{i=n^a}^n i^{\frac{2(\omega-2)}{1-a} -2} n^{-\frac{2a(\omega-2)}{1-a}} = \sum_{i=n^a}^n \frac{1}{i} \cdot i^{\frac{2(\omega-2)}{1-a} -1} n^{-\frac{2a(\omega-2)}{1-a}} .
\end{align*}
Note that 
\begin{align*}
\max_{i\in [n^a,n]} i^{\frac{2(\omega-2)}{1-a} -1} n^{-\frac{2a(\omega-2)}{1-a}} = \max ( n^{a \frac{2(\omega-2)}{1-a} -a} n^{-\frac{2a(\omega-2)}{1-a}} , n^{\frac{2(\omega-2)}{1-a} -1} n^{-\frac{2a(\omega-2)}{1-a}} ) = \max (n^{-a}, n^{2\omega-5}).
\end{align*}
Thus, the second part is
\begin{align*}
\sum_{i=n^a}^n g_i^2 \leq \sum_{i=n^a}^n \frac{1}{i} \cdot \max (n^{-a}, n^{2\omega-5}) = O(\log n) \cdot  \max (n^{-a}, n^{2\omega-5}).
\end{align*}
Combining the first part and the second part completes the proof.
\end{proof}

\subsection{Bounding $v$ move}\label{sec:v_move}

\begin{lemma}[$v$ move]\label{lem:v_move}
We have,
 \begin{align*}
 \sum_{i = 1}^n g_i \cdot \left( \psi(  y^{(k)}_{\pi(i)} ) - \psi ( x^{(k+1)}_{\tau(i)} ) \right) \geq \Omega( \epsilon_{\mathrm{mp}} r_k g_{r_k} / \log n).
\end{align*}
\end{lemma}
\begin{proof}
We first understand some simple facts which are useful in the later proof. 
Note that from the definition of $x^{(k+1)}_i$, we know that $x^{(k+1)}$ has $r_k$ coordinates are $0$ and hence $\| y^{(k)} - x^{(k+1)} \|_0 = r_k$. The difference between those vectors is, for the largest $r_k$ coordinates in $y^{(k)}$, we erase them in $x^{(k+1)}$. Then for each $i \in [n-r_k]$, $x^{(k+1)}_{\tau(i)} = y^{(k)}_{\pi(i+r_k)}$. For convenience, we define $y^{(k)}_{\pi(n+i)} = 0$, $\forall i \in [r_k]$.

We split the proof into two cases.
\paragraph{Case 1.} We exit the while loop when $1.5 r_k \geq n$.

Let $u^*$ denote the largest $u$ s.t.  $| y^{(k)}_{\pi(u)} | \geq \epsilon_{\mathrm{mp}}/2$. 
If $u^* = r_k$, we have that $ |y^{(k)}_{\pi(r_k)} | \geq \epsilon_{\mathrm{mp}}/2 \geq \epsilon_{\mathrm{mp}}/100$.
Otherwise, the condition of the loop shows that
\begin{align*}
 | y^{(k)}_{\pi(r_k)} | \geq ( 1 - 1 / \log n)^{\log_{1.5} r_k - \log_{1.5} u^*} | y^{(k)}_{\pi(u^*)} | \geq ( 1 - 1 / \log n )^{\log_{1.5} n} \epsilon_{\mathrm{mp}}/2 \geq \epsilon_{\mathrm{mp}}/100 .
\end{align*}
where we used that $n \geq 4$. 

According to definition of $x^{(k+1)}_{\tau(i)}$, we have
 \begin{align*}
\sum_{i = 1}^n g_i ( \psi ( y^{(k)}_{\pi(i)} ) - \psi (  x^{(k+1)}_{\tau(i)} ) ) 
= & ~ \sum_{i = 1}^n g_i (\psi ( y^{(k)}_{\pi(i)} ) - \psi( y^{(k)}_{\pi(i+r_k)} ))
\geq \sum_{i = n/3+1}^{n} g_i (\psi ( y^{(k)}_{\pi(i)} ) - \psi( y^{(k)}_{\pi(i+r_k)} )) \\
\geq & ~ \sum_{i = n/3+1}^{n} g_i (\psi ( y^{(k)}_{\pi(i)} ) )
\geq \sum_{i = n/3+1}^{2n/3} g_i \psi ( \epsilon_{\mathrm{mp}} / 100 ) 
\geq \Omega( r_k g_{r_k} \epsilon_{\mathrm{mp}} ), 
\end{align*}
where the first step follows from $x^{(k+1)}_{\tau(i)} = y^{(k)}_{\pi(i+r_k)}$,
the second step follows from the facts that $\psi(|x|)$ is non-decreasing (part 2 of Lemma~\ref{lem:def_psi}) and $|y^{(k)}_{\pi(i)}|$ is non-increasing,
the third step follows from $1.5 r_k > n$ and hence $\psi( y^{(k)}_{\pi(i+r_k)} ) = 0$ for $i \geq n/3+1$,
the fourth step follows from the facts $\psi$ is non-decreasing and $| y^{(k)}_{\pi(i)} | \geq | y^{(k)}_{\pi(r_k)} | \geq \epsilon_{\mathrm{mp}}/100$ for all $i<2n/3$, and the last step follows by the fact $g$ is decreasing and part 3 of Lemma~\ref{lem:def_psi}.

\paragraph{Case 2.}
We exit the while loop when $1.5 r_k < n$ and  $ | y^{(k)}_{\pi(1.5 r_k)} | < (1- 1 / \log n) | y^{(k)}_{\pi(r_k)} | $.

By the same argument as Case 1, we have that $| y^{(k)}_{\pi(r_k)} | \geq \epsilon_{\mathrm{mp}} / 100$.
Part 3 of Lemma~\ref{lem:def_psi} together with the fact
\begin{align*}
| y^{(k)}_{\pi (1.5 r) } | < \min( \epsilon_{\mathrm{mp}}/2, | y^{(k)}_{\pi(r)} | \cdot (1 - 1/\log n)),
\end{align*}
shows that 
\begin{align}\label{eq:r1.5_at_most_r}
\psi( | y^{(k)}_{\pi (1.5 r) } | )  - \psi( | y^{(k)}_{\pi (r) } | ) = \Omega(\epsilon_{\mathrm{mp}} / \log n).
\end{align}

Putting it all together, we have
 \begin{align*}
& ~ \sum_{i = 1}^n g_i \cdot ( \psi( y^{(k)}_{\pi(i)} ) - \psi ( x^{(k+1)}_{\tau(i)} ) ) \\
= & ~ \sum_{i = 1}^n g_i \cdot ( \psi( y^{(k)}_{\pi(i)} ) - \psi ( y^{(k)}_{\pi(i+r_k)} ) ) & \text{~by~} x^{(k+1)}_{\tau(i)} =y^{(k)}_{\pi(i+r_k)} \\
\geq & ~ \sum_{i = r_k/2 }^{ r_k }  g_i \cdot ( \psi ( y^{(k)}_{\pi(i)} ) - \psi ( y^{(k)}_{\pi(i+r_k)} ) ) & \text{~by~}  \psi(y^{(k)}_{\pi(i)}) - \psi(y^{(k)}_{\pi(i+r_k)}) \geq 0 \\
\geq & ~ \sum_{i = r_k/2 }^{ r_k }  g_i \cdot ( \psi ( y^{(k)}_{\pi(r_k)} ) - \psi ( y^{(k)}_{\pi(1.5r_k)} ) )  \\
\geq & ~ \sum_{i= r_k/2 }^{ r_k } g_i \cdot \Omega(\frac{\epsilon_{\mathrm{mp}}}{\log n}) & \text{~by~ \eqref{eq:r1.5_at_most_r}}  \\
\geq & ~ \sum_{i= r_k/2}^{ r_k } g_{r_k} \cdot \Omega(\frac{\epsilon_{\mathrm{mp}}}{\log n}) & \text{~by~}g_i\text{~is~decreasing} \\
= & ~ \Omega\left( \epsilon_{\mathrm{mp}} r_k  g_{r_k} / \log n \right) ,
\end{align*}
where the third step follows by the facts $|y^{(k)}_{\pi(i)}|$ is decreasing and $\psi$ is non-decreasing (from part 2 of Lemma~\ref{lem:def_psi}).
\end{proof}

\subsection{Potential function $\psi$}\label{sec:potential_function_psi}

\begin{figure}[!t]
  \centering
    \includegraphics[width=0.99\textwidth]{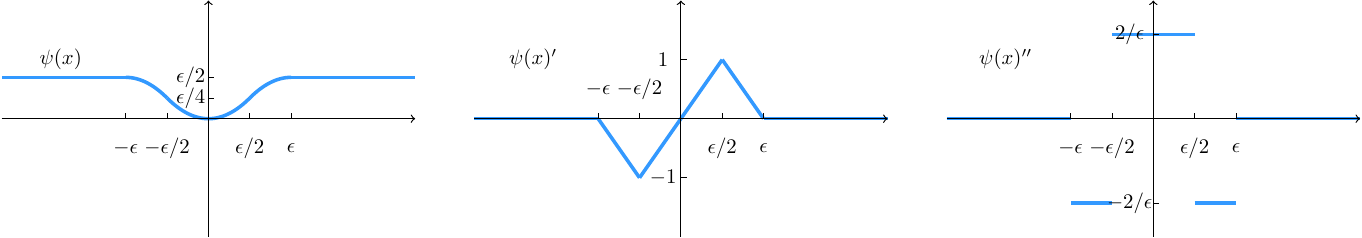}
    \caption{$\psi(x)$, $\psi(x)'$ and $\psi(x)''$. For $\epsilon_{\mathrm{mp}} \in (0,1)$.}
\end{figure}

\begin{lemma}[Properties of function $\psi$]\label{lem:def_psi}
Let function $\psi$ be defined in  \eqref{eq:def_psi}. Then function $\psi$ satisfies the following properties: \\
1. Symmetric $(\psi(-x)=\psi(x))$ and $\psi(0)=0$; \\
2. $\psi(|x|)$ is non-decreasing; \\
3. $|\psi'(x)| = \Omega(1), \forall |x| \in [0.01 \epsilon_{\mathrm{mp}}, \epsilon_{\mathrm{mp}} ]$; \\
4. $L_1 \defeq \max_x \psi'(x) = 1$ and $L_2 \defeq \max_x \psi''(x) = 1 / \epsilon_{\mathrm{mp}}$. \\
\end{lemma}
\begin{proof}
We can see that
\begin{align*}
\psi(x)' = \begin{cases}
\frac{2|x|}{\epsilon_{\mathrm{mp}}}, & |x| \in [0,\epsilon_{\mathrm{mp}}/2] \\
 \frac{ 2(\epsilon_{\mathrm{mp}} - |x|) }{\epsilon_{\mathrm{mp}}}, & |x| \in (\epsilon_{\mathrm{mp}}/2,\epsilon_{\mathrm{mp}}] \\
0, & |x| \in (\epsilon_{\mathrm{mp}}, +\infty)
\end{cases}
\quad\text{and}\quad
\psi(x)'' = \begin{cases}
\frac{2}{\epsilon_{\mathrm{mp}}}, & x \in [0,\epsilon_{\mathrm{mp}}/2] \cup [-\epsilon_{\mathrm{mp}}, -\epsilon_{\mathrm{mp}}/2]\\
-\frac{2}{\epsilon_{\mathrm{mp}}}, & x \in  (\epsilon_{\mathrm{mp}}/2,\epsilon_{\mathrm{mp}}] \cup [-\epsilon_{\mathrm{mp}}/2,0] \\
0. & |x| \in (\epsilon_{\mathrm{mp}}, +\infty)
\end{cases}
\end{align*}
From the $\psi(x)'$ and $\psi(x)''$, it is not hard to see that $\psi$ satisfies the properties needed.
\end{proof}

\subsubsection*{Acknowledgement}

This work was supported in part by NSF Awards CCF-1740551, CCF-1749609, and DMS-1839116. 
We thank S\'{e}bastien Bubeck and Aaron Sidford for helpful discussions. We thank Rasmus Kyng for bringing up the question and providing some fixes in the proof of projection maintenance. We thank Josh Alman for some useful discussions about matrix multiplication. We thank Swati Padmanabhan for writing suggestions. We thank Eric Price for the suggestion of the title of this paper. We thank Shunhua Jiang and Hengjie Zhang for drawing several beautiful pictures. Finally, we thank anonymous STOC and JACM reviewers for their detailed feedback.

\appendix

\addcontentsline{toc}{section}{References}
\bibliographystyle{alpha}
\bibliography{ref}

\section{Appendix}
\begin{lemma}\label{lem:var_xy}
Let $x$ and $y$ denote (possibly dependent) random variables such that $|x| \leq c_x$ and $|y| \leq c_y$ almost surely. Then, we have
$$\Var[xy] \leq 2 c_x^2 \cdot \Var[y] + 2 c_y^2 \cdot \Var[x].$$
\end{lemma}
\begin{proof}

Recall that $\Var[xy]\leq\E[(xy-t)^{2}]$ for any scalar $t$. Hence,
\begin{align*}
\Var[xy] & \leq\E[(xy-\E[x]\E[y])^{2}] =\E[(xy-x\E[y]+x\E[y]-\E[x]\E[y])^{2}]\\
 & \leq2\E[(xy-x\E[y])^{2}]+2\E[(x\E[y]-\E[x]\E[y])^{2}]\\
 & \leq2c_{x}^{2}\cdot\Var[y]+2c_{y}^{2}\cdot\Var[x].
\end{align*}
\end{proof}

\begin{lemma}[\cite{v89b}]\label{lem:classical_step}
Given a matrix $A \in \R^{d \times n}$, vectors $b \in \R^d, c \in \R^n$.
Suppose $x,s,y \in \R^n$ satisfy that $x s \approx_{0.1} t$, $A x = b$ and $A^\top y + s = c$ for some $t>0$. For any $\epsilon \in (0,1/2]$, in $\wt{O}(n^{2.5} \log (n/\epsilon))$ time, we can find vectors $x^{\new}, s^{\new} \in \R^n$ and $y^{\new} \in \R^d$ such that
\begin{align*}
\| x^{\new} s^{\new} - t \|_2 \leq & ~ \epsilon, \\
A x^{\new} = & ~ b, \\
A^\top y^{\new} + s = & ~ c.
\end{align*}
\end{lemma}
\begin{remark}
Instead of using the algorithm in \cite{v89b}, one can also run our algorithm with $k=n$ for $O(\sqrt{n} \log n)$ iterations. Since $k=n$, there is no randomness involved and hence $\Phi$ will decrease deterministically to $O(n)$.
\end{remark}
\begin{lemma}\label{lem:omega_leq_3_minus_a}
$\omega \leq 3 - \alpha$.
\end{lemma}
\begin{proof}
We consider a $n \times n$ matrix $A$ multiply another $n \times n$ matrix $B$. We split $A$ into $n^{1-\alpha}$ fat matrices where each of them has size $n^\alpha \times n$. Since $\omega$ is the best exponent of matrix multiplication, thus we know
\begin{align*}
n^{\omega+o(1)} \leq n^{1-\alpha} \cdot n^{2+o(1)}.
\end{align*}
Taking $n \rightarrow \infty$, this implies $\omega \leq 3 - \alpha$.
\end{proof}
Note that the bound in Lemma~\ref{lem:omega_leq_3_minus_a} can be improved to $\omega + \frac{1}{2} \omega \alpha \leq 3$ \cite{cglz20} via tensor rank.

\begin{lemma}[Rectangular matrix multiplication]\label{lem:rectangular_matrix_multiplication}
For any $n\geq r$, multiplying an $n \times r$ with an $r \times n$ matrix or $n \times n$ with $n \times r$ takes time
$$n^{2+o(1)}+r^{\frac{\omega-2}{1-\alpha}}n^{2-\frac{\alpha(\omega-2)}{1-\alpha}+o(1)}.$$
\end{lemma}
\begin{proof}
The cost for multiplying a $n\times n$ and a $n\times r$ matrix is the same as multiplying a $n\times r$ and a $r\times n$ matrix \cite[page 51]{pan1984multiply}. So, we focus on the later case.

For the case $r \leq n^\alpha$, it follows from the rectangular matrix multiplication result in \cite{gu18}.

For the case $r \geq n^\alpha$, we let $k = (n/r)^\frac{1}{1-\alpha}$. We can view the problem as multiplying a $k\times k^\alpha$ and a $k^\alpha \times k$ block matrices and each block has size $\frac{n}{k} \times \frac{n}{k}$ size. Therefore, the total cost is
$$k^{2+o(1)} \times (\frac{n}{k})^{\omega + o(1)} = r^{\frac{\omega-2}{1-\alpha}}n^{2-\frac{\alpha(\omega-2)}{1-\alpha}+o(1)}.$$
\end{proof}

\begin{lemma}\label{lem:feasible_LP}
Let $A \in \R^{d \times n}$, $b \in \R^d$ and $c \in \R^n$. For a matrix $A$, we define $\| A\|_1$ to be $\sum_{i,j} |A_{i,j}|$. Given a linear program $\min_{Ax = b, x \geq 0} c^\top x$ with $n$ variables and $d$ constraints. Assume that \\
1. Diameter of the polytope : For any $x \geq 0$ with $A x = b$, we have that $\| x \|_{\infty} \leq R$. \\
2. Lipschitz constant of the linear program : $\| c \|_{\infty} \leq L$. 

For any $\delta \in (0,1]$, the modified linear program $\min_{\ov{A} \ov{x} = \ov{b} , \ov{x} \geq 0 } \ov{c}^\top \ov{x}$ with
\begin{align*}
\ov{A}=
\begin{bmatrix}
A & 0 & \frac{1}{R} b - A 1_n \\
1_n^\top & 1 & 0 
\end{bmatrix}
\in \R^{(d+1) \times (n+2)}
,
\ov{b}=
\begin{bmatrix}
\frac{1}{R} b \\
n+1 
\end{bmatrix}
\in \R^{d+1}
\text{~and,~}
\ov{c}
=
\begin{bmatrix}
\frac{\delta}{L} \cdot c \\
0 \\
1
\end{bmatrix}
\in \R^{n+2}
\end{align*}
satisfies the following : \\
1. $\ov{x} = \begin{bmatrix} 1_n \\ 1 \\ 1 \end{bmatrix} \in \R^{n+2}$, $\ov{y} = \begin{bmatrix} 0_d \\ - 1 \end{bmatrix} \in \R^{d+1}$ and $\ov{s} = \begin{bmatrix} 1_n + \frac{\delta}{ L } \cdot c \\ 1 \\ 1 \end{bmatrix} \in \R^{n+2}$ are feasible primal dual vectors. \\
2. For any feasible primal dual vectors $(\ov{x}, \ov{y}, \ov{s}) \in \R^{(n+2) \times (d+1) \times (n+2)}$ with duality gap $\leq \delta^2$, the vector $\wh{x} = R \cdot \ov{x}_{1:n} \in \R^n$ ($\ov{x}_{1:n}$ is the first $n$ coordinates of $x$) is an approximate solution to the original linear program in the following sense
\begin{align*}
c^\top \wh{x} \leq & ~ \min_{A x = b , x \geq 0} c^\top x + LR \cdot \delta, \\
\| A \wh{x} - b \|_1 \leq & ~ 4 n \delta \cdot ( R \| A \|_1 + \| b \|_1 ) , \\
\wh{x} \geq & ~ 0 .
\end{align*}
\end{lemma}
\begin{proof}

{\bf Part 1.} 
For the first result, straightforward calculations show that $(\ov{x}, \ov{y}, \ov{s}) \in \R^{(n+2) \times (d+1) \times (n+2)}$ are feasible, i.e., 
\begin{align*}
\ov{A} \ov{x} = \begin{bmatrix}
A & 0 & \frac{1}{R} b - A 1_n \\
1_n^\top & 1 & 0 
\end{bmatrix} 
\cdot
\begin{bmatrix} 1_n \\ 1 \\ 1 \end{bmatrix}
= 
\begin{bmatrix}
\frac{1}{R} b \\
n+1
\end{bmatrix}
= \ov{b}
\end{align*}
and
\begin{align*}
\ov{A}^\top \ov{y} + \ov{s}
 = & ~
\begin{bmatrix}
A^\top & 1_n \\
0 & 1 \\
\frac{1}{R} b^\top - 1_n^\top A^\top & 0 \\
\end{bmatrix}
\cdot
\begin{bmatrix}
0_d \\
-1 
\end{bmatrix}
+
\begin{bmatrix}
1_n + \frac{\delta}{L} \cdot c \\
1 \\
1
\end{bmatrix} \\
= & ~ 
\begin{bmatrix}
-1_n \\
-1 \\
0\\
\end{bmatrix}
+
\begin{bmatrix}
1_n + \frac{\delta}{ L } \cdot c \\
1 \\
1
\end{bmatrix} \\
= & ~
\begin{bmatrix}
\frac{\delta}{L} \cdot c \\
0 \\
1
\end{bmatrix} \\
= & ~ \ov{c}
\end{align*}

{\bf Part 2.} 
For the second result, we let 
\begin{align*}
\OPT = \min_{Ax = b, x \geq 0} c^\top x, \text{~~~and,~~~} \ov{\OPT} = \min_{\ov{A} \ov{x} = \ov{b} , \ov{x} \geq 0 } \ov{c}^\top \ov{x}.
\end{align*}
For any optimal $x \in \R^n$ in the original LP, we consider the following $\ov{x} \in \R^{n+2}$
\begin{align}\label{eq:another_ov_x_in_appendix}
\ov{x} = 
\begin{bmatrix}
\frac{1}{R} x \\
n+1 - \frac{1}{R} \sum_{i=1}^n x_i \\
0
\end{bmatrix}
\end{align}
and $\ov{c} \in \R^{n+2}$
\begin{align}\label{eq:another_ov_c_in_appendix}
\ov{c} = 
\begin{bmatrix}
\frac{\delta}{L}\cdot c^\top \\
0 \\
1
\end{bmatrix}
\end{align}
We want to argue that $\ov{x} \in \R^{n+2}$ is feasible in the modified LP. It is obvious that $\ov{x} \geq 0$, it remains to show $\ov{A} \ov{x} = \ov{b} \in \R^{d+1}$. We have
\begin{align*}
\ov{A} \ov{x} = 
\begin{bmatrix}
A & 0 & \frac{1}{R} b - A 1_n \\
1_n^\top & 1 & 0 
\end{bmatrix} 
\cdot
\begin{bmatrix}
\frac{1}{R} x \\
n+1 - \frac{1}{R} \sum_{i=1}^n x_i \\
0 
\end{bmatrix}
=
\begin{bmatrix}
\frac{1}{R} A x \\
n+1 
\end{bmatrix}
= 
\begin{bmatrix}
\frac{1}{R} b \\
n+1
\end{bmatrix}
=\ov{b}
,
\end{align*}
where the third step follows from $Ax = b$, and the last step follows from the definition of $\ov{b}$.

Therefore, using the definition of $\ov{x}$ in \eqref{eq:another_ov_x_in_appendix} we have that 
\begin{align}\label{eq:bounding_ov_OPT_by_OPT}
\ov{\OPT} \leq \ov{c}^\top \ov{x} =
\begin{bmatrix}
\frac{\delta}{L} \cdot c^\top & 0 & 1
\end{bmatrix}
\cdot
\begin{bmatrix}
\frac{1}{R} x \\
n+1 - \frac{1}{R} \sum_{i=1}^n x_i \\
0
\end{bmatrix}
= \frac{\delta}{ LR} \cdot c^\top x =  \frac{\delta}{ LR} \cdot \OPT. 
\end{align}
where the first step follows from the fact that modified program is a minimization problem, the second step follows from the definitions of $\ov{x} \in \R^{n + 2}$ \eqref{eq:another_ov_x_in_appendix} and $\ov{c} \in \R^{n + 2}$ \eqref{eq:another_ov_c_in_appendix}, the last step follows from the fact that $x \in \R^n$ is an optimal solution in the original linear program.

Given a feasible $(\ov{x}, \ov{y}, \ov{s}) \in \R^{(n+2) \times (d+1) \times (n+2)}$ with duality gap $\delta^2$. Write $\ov{x} = \begin{bmatrix} \ov{x}_{1:n} \\ \tau \\ \theta \end{bmatrix} \in \R^{n+2}$ for some $\tau \geq 0$, $\theta \geq 0$. We can compute $\ov{c}^\top \ov{x}$ which is $\frac{\delta}{ L } \cdot c^\top \ov{x}_{1:n} + \theta$. Then, we have
\begin{align}\label{eq:eq1_in_appendix}
\frac{ \delta }{ L } \cdot c^\top \ov{x}_{1:n} + \theta \leq \ov{\OPT} + \delta^2 \leq \frac{ \delta }{ LR } \cdot \OPT + \delta^2,
\end{align}
where the first step follows from definition of duality gap, the last step follows from \eqref{eq:bounding_ov_OPT_by_OPT}.

Hence, we can upper bound the $\OPT$ of the transformed program as follows:
\begin{align*}
c^\top \wh{x} = R \cdot c^\top \ov{x}_{1:n} = \frac{LR}{ \delta } \cdot \frac{\delta}{L} c^\top \ov{x}_{1:n} \leq \frac{RL}{\delta} ( \frac{\delta}{LR} \cdot \OPT + \delta^2 ) = \OPT + LR \cdot \delta,
\end{align*}
where the first step follows by $\wh{x} = R \cdot \ov{x}_{1:n}$, the third step follows by \eqref{eq:eq1_in_appendix}.


Note that 
\begin{align}\label{eq:lower_bound_on_delta_L_c_ovx}
\frac{\delta}{L} c^\top \ov{x}_{1:n}\geq -\frac{\delta}{L} \|c\|_\infty \| \ov{x}_{1:n} \|_1 = -\frac{\delta}{L} \|c\|_\infty \| \frac{1}{R} x \|_1 \geq  -\frac{\delta}{L} \|c\|_\infty \frac{n}{R} \| x \|_{\infty} \geq -\delta n,
\end{align}
where the second step follows from the definition of $\ov{x} \in \R^{n+2}$, and the last step follows from $\| c \|_{\infty} \leq L$ and $\| x \|_{\infty} \leq R$.

We can upper bound the $\theta$ in the following sense,
\begin{align}\label{eq:theta_is_at_most_4n_delta}
\theta \leq \frac{\delta}{LR} \cdot \OPT + \delta^2 +\delta n
\leq 2n \delta + \delta^2 \leq  4 n \delta
\end{align}
where the first step follows from \eqref{eq:eq1_in_appendix} and \eqref{eq:lower_bound_on_delta_L_c_ovx}, the second step follows by $\OPT = \min_{A x = b, x \geq 0} c^\top x \leq n L R$ (because $\| c \|_{\infty} \leq L$ and $\| x \|_{\infty} \leq R$), and the last step follows from $\delta \leq 1 \leq n$.

The constraint in the new polytope shows that
\begin{align*}
A \ov{x}_{1:n} + ( \frac{1}{R} b - A 1_n ) \theta = \frac{1}{R} b.
\end{align*}
Using $\wh{x} = R x_{1:n} \in \R^n$, we have
\begin{align*}
A \frac{1}{R} \wh{x} + ( \frac{1}{R} b - A 1_n ) \theta = \frac{1}{R} b.
\end{align*}
Rewriting it, we have $A \wh{x} - b = ( R A 1_n - b ) \theta \in \R^d$ and hence
\begin{align*}
\| A \wh{x} - b \|_1 = \| ( R A 1_n - b ) \theta \|_1 \leq \theta ( \| R A 1_n \|_1 + \| b \|_1 ) \leq \theta \cdot ( R \| A \|_1 + \| b \|_1 ) \leq 4 n \delta \cdot ( R \| A \|_1 + \| b \|_1 ),
\end{align*}
where the second step follows from the triangle inequality, the third step follows from $\| A 1_n \|_1 \leq \| A \|_1$ (because the definition of entry-wise $\ell_1$ norm), and the last step follows from \eqref{eq:theta_is_at_most_4n_delta}.

Thus, we complete the proof.
\end{proof}

\section{Generalized Projection Maintenance}
Given the usefulness of projection maintenance, we state a more general version of Theorem \ref{thm:maintain_projection} for future use.
\begin{theorem}
\label{thm:maintain_projection2}Assume the following about the cost
of matrix operations: 
\begin{itemize}
\item In $O(t_{k})$ time, we can multiply a $n\times n$ and a $n\times k$
matrix, and we can multiply a $n\times k$ and a $k\times n$ matrix.
\item In $O(s_{n})$ time, we can invert a $n\times n$ matrix, and we can
multiply a $n\times n$ and a $n\times n$ matrix.
\item $t_{k}/k$ is decreasing $k$.
\end{itemize}
Given a matrix $A\in\R^{d\times n}$ with $n\geq d$, a tolerance
parameter $0<\epsilon_{\mathrm{mp}}<1/4$ and $k^{*}\in[n]$, there is a deterministic
data structure that approximately maintains the projection matrices
\begin{align*}
\sqrt{W}A^{\top}(AWA^{\top})^{-1}A\sqrt{W}
\end{align*}
and the inverse matrices
$(A W A^{\top})^{-1} \in \R^{d \times d}$ for positive diagonal matrices $W \in \R^{n \times n}$ through the
following operations: 
\begin{itemize}
\item $\textsc{Update}(w)$: Output a vector $\tilde{v}$ such that for
all $i$, 
\begin{align*}
(1-\epsilon_{\mathrm{mp}})\tilde{v_{i}}\leq w_{i}\leq(1+\epsilon_{\mathrm{mp}})\tilde{v_{i}}.
\end{align*}
\item $\textsc{Query}_{1}(h)$: Output $\sqrt{\tilde{V}}A^{\top}(A\tilde{V}A^{\top})^{-1}A\sqrt{\tilde{V}}h \in \R^n$
for the $\tilde{v}$ outputted by the last call to $\textsc{Update}$. 
\item $\textsc{Query}_{2}(h)$: Output $(A\tilde{V}A^{\top})^{-1}h \in \R^d$ for
the $\tilde{v}$ outputted by the last call to $\textsc{Update}$.
\item $\textsc{Insert}(a,w_{a})$: Insert a column $a$ into $A$, a weight
$w_{a}$ into $w$.
\item $\textsc{Delete}(a)$: Delete a column $a$ from $A$ and its
corresponding weight from $w$.
\item $\textsc{Output}()$: Output $\sqrt{\tilde{V}}A^{\top}(A\tilde{V}A^{\top})^{-1}A\sqrt{\tilde{V}} \in \R^{n \times n}$
and $(A\tilde{V}A^{\top})^{-1} \in \R^{d \times d}$ for the $\tilde{v}$ outputted by
the last call to $\textsc{Update}$. 
\end{itemize}
Suppose that the number of columns is $O(n)$ during the whole algorithm
and that for any call of $\textsc{Update}$, we have 
\begin{align*}
\sum_{i=1}^{n}\left( \E[ \ln w_{i} ] - \ln ( w_{i}^{\mathrm{old}} ) \right)^{2}\leq C_{1}^{2}\qquad\text{and}\qquad\sum_{i=1}^{n} ( \Var[\ln (w_{i} ) ] )^{2} \leq C_{2}^{2}
\end{align*}
where $w$ is the input of call, $w^{\mathrm{old}}$ is the weight
before the call, and the expectation and variance is conditional on
$w_{i}^{\mathrm{old}}$. Then, we have that:
\begin{itemize}
\item The data structure takes $O(s_{n}+t_{n})$ time to initialize.
\item Each call of $\textsc{Query}$ takes time $O(n\cdot\|h\|_{0}+s_{k^{*}}+nk^{*})$.
\item Each call of $\textsc{Output}()$ takes $O(t_{k^{*}})$ time.
\item Each call of $\textsc{Insert}$ and $\textsc{Delete}$ takes $O(n^{2})$
time.
\item Each call of \textsc{Update} takes 
\begin{align*}
O \left( (C_{1}/\epsilon_{\mathrm{mp}}+C_{2}/\epsilon_{\mathrm{mp}}^{2})\cdot \left( \frac{t_{k^{*}}^2}{k^{*}}+\sum_{i=k^{*}}^{n}\frac{t_{i}^{2}}{i^{2}} \right)^{1/2} \cdot\log n \right)
\end{align*}
expected time in amortized.
\end{itemize}
\end{theorem}

\begin{proof}
The proof is essentially the same as Theorem \ref{thm:maintain_projection}. The way to maintain
$A\tilde{V}A^{\top} \in \R^{d \times d}$ is almost identical to $\sqrt{\tilde{V}}A^{\top}(A\tilde{V}A^{\top})^{-1}A\sqrt{\tilde{V}}h \in \R^n$.
Updating both matrices under insertion and deletion can be done via
Sherman Morrison formula in $O(n^{2})$ time. We note that these updates
does not increase our potential and hence it does not affect the amortized
cost for $\textsc{Update}$. Finally, to bound the runtime using $t_{k}$
and $s_{n}$ instead of $\alpha$ and $\omega$, we use the same potential
$\Psi_{k}=\sum_{i=1}^n g_{i}\psi(x_{i}^{(k)})$ with a new definition of
$g$: 
\begin{align*}
g_{i}=\begin{cases}
t_{i}/i, & \text{if }i\geq k^{*};\\
t_{k^{*}}/k^{*}, & \text{otherwise}.
\end{cases}.
\end{align*}
Now, the update time in Lemma \ref{lem:maintain_projection_update} becomes $O(rg_{r})$. The rest of
the proof is identical. In particular, Lemma \ref{lem:sum_g_i} becomes
\begin{align*}
\|g\|_{2}= \left( \frac{t_{k^{*}}}{k^{*}}+\sum_{i=k^{*}}^{n}\frac{t_{i}^{2}}{i^{2}} \right)^{1/2}
\end{align*}
and hence Lemma \ref{lem:bounding_E_psi_y_minus_psi_x} gives the bound 
\begin{align*}
O (C_{1}+C_{2}/\epsilon_{ \mathrm{mp} } ) \cdot \left( \frac{t_{k^{*}}}{k^{*}}+\sum_{i=k^{*}}^{n}\frac{t_{i}^{2}}{i^{2}} \right)^{1/2}.
\end{align*}
\end{proof}


\begin{figure}[!h]
    \centering
    \includegraphics[width=0.98\textwidth]{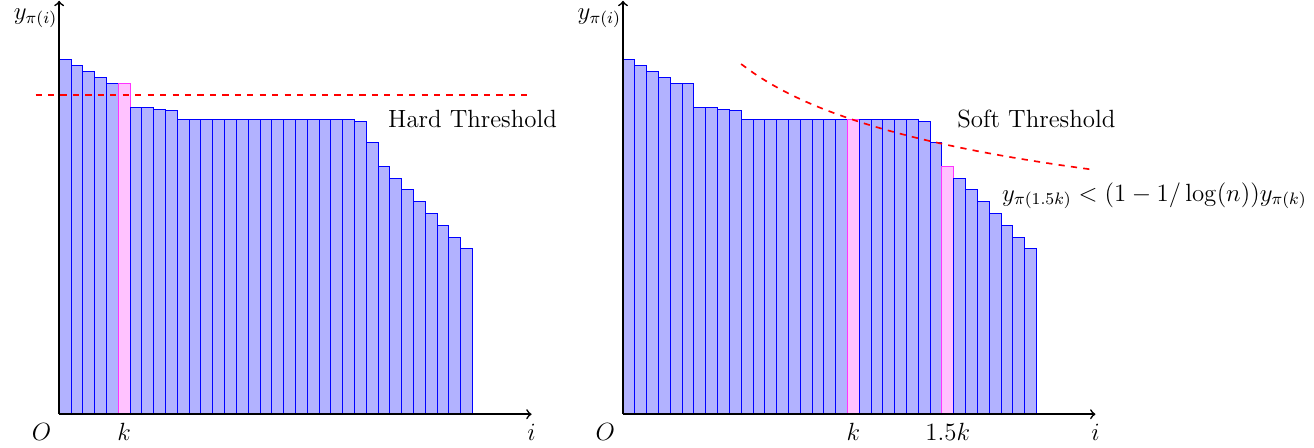}
    \caption{In this figure we illustrate the naive hard threshold and the soft threshold. The x-axis represents the sorted $n$ coordinates, and the y-axis represents the errors $y_{\pi(i)}$. All the coordinates smaller than the threshold $k$ are updated. In the left figure, we choose the hard threshold $k$ as the smallest coordinates such that $y_{\pi(k)}\leq \epsilon_{\mathrm{mp}}$. In the right figure, we choose the soft threshold $k$ as the smallest coordinates that satisfies both $y_{\pi(k)}\leq \epsilon_{\mathrm{mp}}$ and $y_{\pi(1.5k)}<(1-1/\log(n))y_{\pi(k)}$.}
    \label{fig:hard_soft_thresholding}
\end{figure}

%
%
%
%




\end{document}